\documentclass{article}
\usepackage[T1]{fontenc}
\usepackage[margin=1in]{geometry}
\usepackage{amsfonts,amsmath,amsthm,amssymb,mathtools}  

\mathtoolsset{centercolon}
\usepackage{xfrac,nicefrac}
\usepackage{mathdots}
\usepackage{mleftright}  
\let\left\mleft
\let\right\mright

\usepackage{xspace}
\xspaceaddexceptions{]\}}  
\usepackage{regexpatch}

\usepackage{bm,bbm,dsfont}  
\usepackage{caption}
\usepackage[normalem]{ulem}
\usepackage{enumitem}

\usepackage{graphicx}
\usepackage{float}
\usepackage{subcaption}  
\usepackage{tcolorbox}
\usepackage{tikz}
\usetikzlibrary{decorations.pathreplacing}
\usetikzlibrary{calc}
\usetikzlibrary{positioning}
\usetikzlibrary{arrows.meta}

\usepackage[linesnumbered,boxed,ruled,vlined]{algorithm2e}

\SetCommentSty{mycommfont}

\usepackage{thmtools,thm-restate}
\usepackage[colorlinks,citecolor=blue,linkcolor=blue,urlcolor=red]{hyperref}

\theoremstyle{plain}

\newtheorem{theorem}{Theorem}[section]  
\newtheorem{lemma}[theorem]{Lemma}

\newtheorem{corollary}[theorem]{Corollary}
\newtheorem{invariant}[theorem]{Invariant}

\newtheorem{claim}[theorem]{Claim}

\theoremstyle{definition}  

\newtheorem{remark}[theorem]{Remark}



\usepackage[capitalise]{cleveref}
\crefname{algocf}{Algorithm}{Algorithms}
\Crefname{algocf}{Algorithm}{Algorithms}
\crefname{claim}{Claim}{Claims}
\Crefname{claim}{Claim}{Claims}
\crefname{invariant}{Invariant}{Invariants}

\newfloat{Distribution}{htbp}{loa}
\crefname{Distribution}{Distribution}{Distributions}
\Crefname{Distribution}{Distribution}{Distributions}
\newfloat{Protocol}{htbp}{loa}
\crefname{Protocol}{Protocol}{Protocols}
\Crefname{Protocol}{Protocol}{Protocols}

\SetKwProg{Function}{Function}{:}{}

\SetKwProg{CodeBlock}{}{}{}
\SetKwProg{Repeat}{Repeat}{:}{}

\DeclarePairedDelimiter{\bk}{(}{)}
\DeclarePairedDelimiter{\Bk}{[}{]}

\DeclarePairedDelimiterX\mysetbase[2]{\lbrace}{\rbrace}{#1\,\delimsize\vert\,#2}
\NewDocumentCommand{\myset}{sO{}m m}{%
  \IfBooleanTF{#1}
    {\mysetbase*{#3}{#4}}
    {\mysetbase[#2]{#3}{#4}}
}

\DeclareMathOperator*{\E}{\mathbb{E}}

\let\Pr\PrAux
\DeclareMathOperator{\poly}{poly}

\renewcommand{\tilde}{\widetilde}

\newcommand{\eps}{\varepsilon}

\renewcommand{\l}{\ell}
\renewcommand{\emptyset}{\varnothing}
\renewcommand{\epsilon}{\eps}

\newcommand{\defn}[1]{\emph{\boldmath\textbf{#1}}}

\usepackage{regexpatch}
\makeatletter
\xpatchcmd\thmt@restatable{%
\csname #2\@xa\endcsname\ifx\@nx#1\@nx\else[{#1}]\fi
}{%
\ifthmt@thisistheone
\csname #2\@xa\endcsname\ifx\@nx#1\@nx\else[{#1}]\fi
\else
\csname #2\@xa\endcsname[{Restated}]
\fi}{}{}
\makeatother

\newcommand{\Value}{value}

\title{Resizable Retrieval}
\author{
William Kuszmaul%
\thanks{Supported in part by NSF grant CNS2504471 and by a Jane Street Research Grant. \texttt{kuszmaul@cmu.edu}.}\\
CMU
\and
Aaron Putterman%
\thanks{Supported in part by the Simons Investigator Awards of Madhu Sudan and Salil Vadhan, AFOSR award
FA9550-25-1-0112, and a Jane Street Graduate Research Fellowship. \texttt{aputterman@g.harvard.edu}.}\\
Harvard University
\and
Tingqiang Xu%
\thanks{\texttt{xtq23@mails.tsinghua.edu.cn}.}\\
Tsinghua University
\and
Hangrui Zhou%
\thanks{\texttt{zhouhr23@mails.tsinghua.edu.cn}.}\\
Tsinghua University
\and
Renfei Zhou%
\thanks{Supported in part by NSF grant CNS2504471, the Jane Street Graduate Research Fellowship, and the MongoDB PhD Fellowship. \texttt{renfeiz@andrew.cmu.edu}.}\\
CMU
}
\date{}

\begin{document}

\maketitle

\begin{abstract}
A dynamic retrieval data structure encodes a function $f:\mathcal{K} \rightarrow [2^v]$ for a set $\mathcal{K} \subseteq [U]$, while supporting queries $f(x)$ for $x\in \mathcal{K}$, insertions \texttt{Insert}$(x, f(x))$ for $x \notin \mathcal{K}$, and deletions \texttt{Delete}$(x)$ for $x \in \mathcal{K}$. Given an upper bound $N$ on $|\mathcal{K}|$, it is known how to solve the dynamic retrieval problem with $O(1)$-time operations and space $Nv + O(N \log \log (U/N))$ bits. An open question, first posed by Demaine et al.\ in 2006, is whether a similar bound can be achieved with a resizable data structure, whose space bound is parameterized by the \emph{current} size $n$ of $\mathcal{K}$. We answer this question in the affirmative and prove matching lower bounds for the space-time trade-off achieved by our data structure. We also give corollaries for space-efficient memory allocation and dynamic filters.
\end{abstract}

\section{Introduction}
\label{sec:intro}

A retrieval data structure is a space-efficient data structure that encodes a function $f: \mathcal{K} \rightarrow [2^v]$ mapping a set $\mathcal{K} \subseteq [U]$ of keys to $v$-bit values. The data structure supports a \texttt{Query} function that, on any key $k \in \mathcal{K}$, returns $f(k)$. If a user calls \texttt{Query}$(k)$ on a key $k \notin \mathcal{K}$, then the function is permitted to return an arbitrary $v$-bit value. 

What makes retrieval data structures interesting (as opposed to, say, dictionaries) is that they can be implemented using very little space. In the static setting, where the set $\mathcal{K}$ and the function $f$ are fixed, the retrieval problem can be solved using as little as $nv + o(n)$ bits. This is almost the same amount of space that would be needed just to store an array of $n$ $v$-bit values. 

The reason that retrieval data structures can be so space efficient is that they do not, in principle, need to store as much information about the actual key set $\mathcal{K}$. So long as the data structure returns the correct value $f(k)$ for each key $k \in \mathcal{K}$, the data structure need not actually know which keys $k$ are or are not in $\mathcal{K}$. 

\emph{A priori}, the fact that retrieval data structures store so little information about the key set $\mathcal{K}$ might seem to suggest that space-efficient retrieval is only possible in the static setting. However, a celebrated 2006 paper by Demaine, Meyer auf der Heide, Pagh, and P{\v a}tra{\c s}cu \cite{demaine2006dictionariis} showed that, at least in some sense, this is not the case. They considered a setting in which one adds two operations to the data structure:
\begin{itemize}
\item \textbf{Insert($x, f(x)$):} This inserts a new element $x \not\in \mathcal{K}$ into $\mathcal{K}$ and sets its value $f(x)$. This operation requires that $x \not\in \mathcal{K}$. 
\item \textbf{Delete($x$):} This removes an element $x \in \mathcal{K}$ from $\mathcal{K}$. This operation requires that $x \in \mathcal{K}$ before calling \textbf{Delete($x$)}.
\end{itemize}

Demaine, Meyer auf der Heide, Pagh, and P{\v a}tra{\c s}cu \cite{demaine2006dictionariis} showed that, given an upper bound $N$ on $|\mathcal{K}|$, it is possible to support both insertions and deletions of keys/values, while using total space 
$$Nv + O(N \log \log (U/N))$$
bits, and while supporting $O(1)$-time operations with high probability. This space bound turns out to be information-theoretically optimal \cite{MPP05, RetrievalPhase}.

What is not clear is whether it might be possible to achieve an even stronger space guarantee, where the space usage of the data structure is parameterized not by some known upper bound $N$ on the size of $\mathcal{K}$, but rather by whatever the current size $n = |\mathcal{K}|$ of the set is at any given moment. Whether such a resizing guarantee is possible was left as the main open question in \cite{demaine2006dictionariis}.\footnote{\emph{A priori}, one might wonder: Why can't we build a resizable retrieval data structure by simply rebuilding the data structure, whenever the size changes by a significant amount? What prevents this is that the data structure does not have the ability to \emph{recall} the keys $k$ in $\mathcal{K}$. It must perform resizing with limited information.}

In the nearly two decades since Demaine et al.'s \cite{demaine2006dictionariis} result, the retrieval problem has continued to be the subject of a great deal of interest. Subsequent work includes both upper \cite{DW19,Por09,DP08,monotone1,KW24} and lower \cite{MPP05,RetrievalPhase,monotone2,monotone3,KW24} bounds, with work on time efficiency in the static setting \cite{DW19,Por09,DP08}, on middle grounds between the static and dynamic settings \cite{KW24,RetrievalPhase}, on the special case where $f$ encodes element ranks \cite{monotone1, monotone2, monotone3,monotone4}, on relaxations in which queries are permitted to incur errors \cite{chazelle2004bloomier,maplets,infinifilter}, and on applications to other data structures \cite{bender2025optimal,BCFKT23,GL22,DW21,Por09,DP08}. 

Recently, Bercea and Even \cite{bercea2022extendable} were able to make the first concrete progress on Demaine et al.'s resizability question \cite{demaine2006dictionariis}. Building, in part, on an earlier line of work on extendable filters \cite{PSW13,liu2020succinct}, they give a retrieval data structure achieving a space bound of the form $n_{\text{max}} v + O(n_{\text{max}} \log \log U)$, where $n_{\text{max}}$ is the largest size that the data structure has ever had so far.\footnote{We remark that, due to the slightly informal nature of the discussions in Section 5 of \cite{PSW13} and Section 5.2 of \cite{bercea2022extendable} a reader may get the impression that, for both filters and retrieval, the known constructions naturally extend to support space guarantees that scale with the \emph{current} number of keys (rather than the all-time \emph{maximum}). However, the constructions in these papers (which, in the insertion-only setting, use essentially the same setup as the constructions in this paper) are unable to shrink the table as elements are deleted. In particular, when the data structure is resized to be smaller, this introduces \emph{new} collisions between keys, and these new collisions are very different from those created by insertions, because the data structure no longer contains \emph{either} key involved in the collision. Our construction resolves this issue by introducing a new piece of the data structure,  the ``mid-yard'', which stores colliding keys for which the true key is not known. Supporting the midyard in a space and time efficient manner is exactly what constitutes the difficulty of the current work, and is our main contribution to the area.} In fact, they give an even stronger guarantee, producing a method for assigning keys to (stable!) hash codes in the range $[(1 + o(1)) \cdot n_{\text{max}}]$. Such an assignment is not, in general, feasible if one wishes to parameterize by $n$ (the current number of keys) rather than $n_{\text{max}}$.

In the setting where we wish to support both insertions and deletions, and achieve a space bound parameterized by $n$ (rather than $n_{\text{max}})$, the question of whether a resizable retrieval data structure is possible \cite{demaine2006dictionariis} has remained open.

\paragraph{This paper: Efficient resizable retrieval.} The main technical contribution of this paper is an affirmative answer to Demaine et al.'s \cite{demaine2006dictionariis} question:

\begin{restatable}{theorem}{mainthm}
  \label{thm:mainRestated}
  Let $U$ and $v$ be parameters satisfying $v \le O(\log U)$, let $k = O(1)$, and let $\eps >0$ be arbitrary constants. There is a \emph{resizable} retrieval data structure maintaining a set of key-value pairs, with keys from the universe $[U]$ and $v$-bit values, such that:
  \begin{enumerate}
      \item At any given moment when it stores $n$ keys, with $\mathrm{polylog}(U) \leq n\leq U/2$, it occupies $O(n \log \log (U / n) + n\log^{(k)}n) + nv + U^{\eps}$ bits of space.
      \item It supports insertions, deletions, and retrieval queries in $O(k) = O(1)$ time with high probability in $n$ under the word RAM model with word size $w = \Theta(\log U)$.
  \end{enumerate}
\end{restatable}

We remark that the $U^\epsilon$ space term in Theorem \ref{thm:mainRestated} is simply for storing highly random hash functions. If we assume access to fully random hash functions, then the term can be replaced with a poly-logarithmic dependence on $U$. Our assumption that $n \geq \mathrm{polylog}(U)$ simply means that $n \geq \lg^C(U)$, for some sufficiently large constant $C$. This assumption is very mild, and is only used to ensure that certain concentration bounds hold during deamortization.

{In addition to proving Theorem \ref{thm:mainRestated}, we also prove an extension of the theorem that allows different values to have different sizes. Specifically, the extension allows the value for each key $x$ to have an arbitrary length of $v_x \in [O(\log U)]$ bits, where the length $v_x$ is specified on insertion. Supposing $U \ge \poly n$, this version of the theorem replaces the $nv$ term in the space usage of Theorem \ref{thm:mainRestated} with a term of the form $\sum_x v_x$, where $x$ sums over the current key set. One can view this extension as a very space-efficient memory allocator for small objects, allowing a user (or data structure) to allocate and deallocate objects with fine-grained sizes, without incurring the time/space/pointer overheads classically associated with memory allocation \cite{luby1996tight,bender2017cost,bender2024tiny}. Such an allocator has the potential to significantly simplify the task of performing memory management when designing a space-efficient data structure.
}

The proof of \cref{thm:mainRestated} is spread across \cref{sec:inefficient}, \cref{sec:algorithm}, and \cref{sec:hash_function}.  We begin in \cref{sec:inefficient} with a time-inefficient solution that attempts to distinguish different keys in $\mathcal{K}$ by space-efficiently storing distinct fingerprints for (almost all of) them. A basic challenge here is that, depending on when a key is inserted, the length of the fingerprint that we have for that key may vary (and we have no way of extending it!). This means that keys which were inserted when the data structure was very small may have a high ``collision risk’’ with other keys once the data structure gets larger. Fortunately, if we select the fingerprints in the right way (extending a technique that was previously used in \cite{PSW13} as well as \cite{bercea2022extendable}), then we can amortize the ``collision risk’’ of each key based on when the key was inserted out of all the current keys. 

A more significant challenge is supporting fingerprint-size \emph{reduction} when the filter shrinks (this is the main challenge that does not appear in the incremental setting \cite{PSW13,liu2020succinct,bercea2022extendable}). Here, we must be careful to avoid the introduction of new collisions. We handle this by adding an extra layer to our data structure (called the ``midyard'') that stores extra-long fingerprints for certain keys in order to avoid these shrinkage-based collisions. We remark that, although the midyard allows us to get our desired correctness and space guarantees, it also ends up being one of the most difficult parts of the data structure to make time efficient, due to the fact that fingerprints in the midyard do not all have the same lengths as one another.

The most interesting part of the construction appears in \cref{sec:algorithm}, where we show how to take our time-inefficient solution and turn it into a time-efficient one. Here, we show how to encode the query operations to the data structure as a combination of query operations to hash tables, and query operations to predecessor/successor data structures. A priori, this would seem like \emph{bad news}, since predecessor/successor data structures tend to require $\omega(1)$ time per operation both to query and modify. However, we are able to design the data structure in such a way that two things occur: (1) the predecessor/successor data structures are (essentially) static for large swaths of time; and (2) the predecessor/successor data structures are all in one of two regimes: the ``very sparse’’ regime\footnote{This means that the universe size for the predecessor data structure is at least quasi-exponential in the number of elements.}, where fusion node \cite{fusion} techniques can be used to get time-efficient operations; or (2) the ``very dense’’ regime\footnote{This means that the universe size for the predecessor data structure is near-linear in the number of elements.}, where the succincter-type data structures \cite{4690964} can be used to support fast operations without compromising space efficiency. With a careful analysis, as well as a de-amortization argument, we are able to achieve constant time for all operations while (mostly) preserving the space efficiency of our time-inefficient solution. Finally, in \cref{sec:hash_function}, we show how to implement our data structure without the assumption that the keys are chosen uniformly at random from our universe.

\paragraph{Two lower bounds.}
Finally, we complement Theorem \ref{thm:mainRestated} with two lower bounds, each of which follows from a reduction from another well-studied problem \cite{PSW13,10353191}. 

The first lower bound concerns the $n \log^{(k)} n$ term in the space bound in Theorem \ref{thm:mainRestated}. If one does not care about time efficiency, then this term can be eliminated. We prove, however, that if one wants $O(1)$-time operations, then the term is unavoidable (see Theorem \ref{thm:spaceTime}).

The second lower bound considers the $\log \log (U/n)$ term in the space bound in Theorem \ref{thm:mainRestated}. We show that this term is necessary, not just for Theorem \ref{thm:mainRestated}, but even if we want a weaker version of the theorem that supports insertions but not deletions (see Theorem \ref{thm:lowerinsonly}). This contrasts a recent result by Kuszmaul et al.~\cite{RetrievalPhase} which shows that in the fixed-capacity setting (where we have an upper bound $N$ on the maximum number of keys), and when $v = \Theta(\log n)$,  the insertion-only case can be solved using space $nv + O(n)$ bits. 

\paragraph{A corollary for resizable filters. } Finally, as an immediate corollary of our upper-bound constructions, we can obtain improved bounds for the related problem of constructing a \defn{resizable filter}.

A dynamic filter with false-postive rate $\epsilon$ is a data structure that supports insertions and deletions on some set $S$, while supporting an approximate membership query function $\texttt{Query}(x)$, which returns true for any $x \in S$, and which returns false with probability at least $1-\epsilon$ for any $x \not\in S$. 

If a filter has a fixed capacity $n$, then the optimal space bound is $n \log \epsilon^{-1} + \Theta(n)$ bits \cite{Por09,KW24,kuszmaul2025fingerprint}, and with relatively minor constraints on $\epsilon$, this bound can be combined with constant-time operations \cite{10.1145/3519935.3519969,arbitman2010backyard,bercea2020dynamic}. For resizable filters, where $n$ is not known ahead of time, it is known that any solution must use $n \log \epsilon^{-1} + \Omega(n \log \log (U/n))$  bits \cite{PSW13}. On the upper bound side, supposing that $n \ge U^{\Omega(1)}$, it is known how to construct a solution using 
\begin{equation}
n_{\text{max}} \log \epsilon^{-1} + O\left(n_{\text{max}} \log \log n_{\text{max}} \right)
\label{eq:spacegood}
\end{equation}
bits of space, where $n_{\text{max}}$ is the \emph{largest} number of items ever present, that supports constant amortized-expected time operations with high probability. It was also subsequently shown \cite{liu2020succinct} how to achieve the same bound with high-probability constant-time operations, so long as the filter is insertion-only. 

As an immediate corollary of Theorem \ref{thm:mainRestated}, we get a solution, supporting insertions and deletions, that uses space
\begin{equation} n\log \epsilon^{-1} + O(n \log \log n)
\label{eq:spaceverygood}
\end{equation}
bits, where $n$ is the \emph{current} number of keys. As in \eqref{eq:spacegood}, we require $n \ge U^{\Omega(1)}$ for \eqref{eq:spaceverygood} to hold. Our solution supports constant time operations with high probability in $n$. This corollary is provided formally as \cref{cor:resizableFilter}.

It remains an interesting question whether one might be able to directly extend the constructions in \cite{PSW13} or \cite{liu2020succinct} to get a bound of the form \eqref{eq:spaceverygood}. Note that, if one restricts to insertions only (no deletions), and focuses on constructing a filter rather than a retrieval data structure, then the approach taken in \cite{PSW13} (in their Construction 2) makes use of a similar idea as in the current paper, namely the idea of a ``future fingerprint'' used in Section \ref{sec:inefficient} (the same idea was also used in \cite{bercea2022extendable}). The approach taken in \cite{liu2020succinct}, on the other hand, is quite different from ours, even in the insertion-only setting. An interesting feature of the approach in \cite{liu2020succinct} is that they develop a set of techniques (extended ideas from \cite{bender2018bloom}) for, under certain conditions, directly indexing a set of fingerprints of different lengths. It would be interesting to see if one could use a similar approach to handle the ``midyard'' fingerprints in the current paper, as it would represent an entirely different approach than the one taken in Section \ref{sec:midyard}.

\section{Our Algorithm on a Trie}
\label{sec:inefficient}

We start by presenting a time-inefficient data structure for resizable retrieval which yields the desired space consumption. Without loss of generality, we assume $U$ is a power of $2$ in order to present our algorithm more clearly. In this section, we will also assume that all \emph{distinct} inserted/queried keys are uniformly random from the universe $[U]$. Note that this does not mean keys in all operations are independent because insertions of previously deleted keys are allowed. Later in \cref{sec:hash_function}, we will remove this assumption by applying a global hash function $H$ on each input key, costing only $U^\epsilon$ extra bits of space.

\paragraph{View of the Trie.}

Each key $x$ can be viewed as a random string of $\log U$ bits. In our data structure, we will store a prefix of each key $x$, which is denoted by $\tilde{x}$ and called the \defn{fingerprint} of $x$, and discard the remaining bits of $x$.

In the presentation of our data structure, we will view our finerprints as living in a Trie $T$ of depth $\log U$ whose nodes have a one-to-one correspondence with the bit strings of length $\leq \log U$. Going forward, we will directly refer to bit strings of length $\leq \log U$ by their corresponding nodes on the Trie (and vice versa). 
Each key $x$ is represented by a leaf node on the Trie, and its fingerprint $\tilde{x}$ is represented by an ancestor of that leaf. In the data structure, we will only store the information that the key is within the subtree of $\tilde{x}$, but we will not know which leaf node $x$ is exactly.

In our algorithm, we will maintain the invariant that all the stored fingerprints on the Trie do not form ancestor-descendant relationships. Therefore, when we query some key $y$, there will be at most one fingerprint $\tilde{x}$ on the Trie that is $y$'s ancestor, and thus the data structure can simply return $\tilde{x}$'s associated value without the risk of being incorrect.
However, a small fraction of the keys will fail to satisfy this invariant, in which case we explicitly store those keys in a backyard (i.e., we do not discard bits from those keys, and they do not appear on the Trie).

In addition to the fingerprint, we also store a \defn{future fingerprint} $x_f$ for each key $x$, which is the next few bits of $x$ after the fingerprint (prefix) $\tilde{x}$.
This allows the data structure to (slightly) increase the length of the fingerprint $\tilde{x}$ in the future if necessary. Note that, although the future fingerprint is also ``known information'' about the key $x$, we do not represent it in the view of the Trie, i.e., each Trie node still only represents the fingerprint part $\tilde{x}$ of each key. The future fingerprints will be stored together with the associated values for the keys. The idea of a future fingerprint has also been used in previous work in the extendable setting \cite{PSW13,bercea2022extendable}.

\paragraph{Storage structure.}

Our data structure defines two thresholds: the \defn{standard fingerprint length} $\l = \log n + 10\log\log(U/n) + c$ 
and the \defn{standard future fingerprint length} $\l_f = 10\log\log(U/n) + c$, where $c$ is just some large constant.\footnote{We assume for clarity that $\l$ and $\l_f$ are integers, otherwise one can take the ceiling of these values.} They are both determined by the current number of keys $n$ stored in the data structure.
We rebuild the entire data structure after every $n/100$ operations, possibly changing $\l$ and $\l_f$ due to the change of $n$. The thresholds $\l$ and $\l_f$ can only be changed at the rebuilds: Between two rebuilds, the ideal values of $\l$ and $\l_f$ will change by at most 1, and thus we view them as fixed thresholds that do not change over time (until the next rebuild comes). The reason the values can only change by at most $1$ is that for any $n\in [1, U/2], n'\in [0.99n, 1.01n]$,
$$(\log n'+10\log\log(U/n'))-(\log n+10\log\log(U/n))=\log\frac{n'}{n}+10\log\frac{\log(U/n')}{\log (U/n)}.$$
 Note that $\log\frac{n'}{n}\in [\log 0.99, \log 1.01]\subset [-0.1, 0.1]$, and
$$\log\frac{\log(U/n')}{\log(U/n)}=\log\left(1-\frac{\log(n'/n)}{\log(U/n)}\right)\in[\log\left(1-\log 1.01\right), \log\left(1-\log 0.99\right)]\subset [-0.09, 0.09]$$
since $\log(U/n)\geq 1$. Thus the ideal values will change by at most $0.1+10\cdot 0.09=1$.

We divide the Trie $T$ into two parts -- vertices with depth $\leq \l$, which we call the \defn{frontyard}, and vertices with depth $> \l$, which we call the \defn{midyard}. There is also a \defn{backyard} in our data structure, which includes the keys that we explicitly store. As we discussed earlier, the keys in the backyard are not represented on the Trie. As we will show, most of the keys will be stored in the frontyard, which means that we store a fingerprint of at most $\l$ bits for each of them, while a small fraction of keys will be stored in the midyard or the backyard. 
Our data structure will store the keys in the three parts separately as follows.
\begin{itemize}
\item For each key $x$ in the frontyard, we store a tuple $(\tilde{x}, x_f, \Value)$ for it, where $\tilde{x}$ and $x_f$ represent the fingerprint and the future fingerprint of the key as we discussed above, and $\Value$ is the associated value of the key.
\item For each key $x$ in the midyard, we only store the fingerprint $\tilde{x}$ and its associated value, but there is no future fingerprint for $x$.
\item For each key $x$ in the backyard, we explicitly store the entire key $x$ with its associated value.
\end{itemize}

\paragraph{Operations.}

Now, we introduce the procedures for performing queries and updates to our data structure.

\begin{itemize}
    \item To query a key $x$, we first check if $x$ is explicitly stored in the backyard, and return its associated value if this is the case. Otherwise, since no fingerprint is an ancestor of another one on the Trie, we can always find the unique fingerprint in the data structure that is a prefix of the queried key (if there exists one), and return its associated value.
    \item To insert a key $x$, we first try to insert it to the frontyard with an $\l$-bit fingerprint $\tilde{x}$ (and $\l_f$ bits of future fingerprint). This fails when one of the descendants or ancestors of $\tilde{x}$ is already a fingerprint in the data structure, in which case we directly insert $x$ to the backyard.
    \item Deleting an element $x$ is also straightforward: We first call the query procedure to find $x$ in the data structure. No matter whether $x$ was explicitly stored in the backyard or stored on the Trie as a fingerprint $\tilde{x}$, we simply remove it from the data structure.
\end{itemize}

If the standard fingerprint length $\l$ does not change, and we never rebuild, then the midyard will be empty, and all fingerprints in the data structure will have lengths exactly $\l$. In this case, we do not need the future fingerprint, and the data structure degenerates to the previously known data structure in \cite{demaine2006dictionariis} with a space consumption of
\[O\bk*{n\log {\frac {2^\l}n} + \frac {n^2}{2^\l} \log\bk*{\frac{U\cdot 2^\l}{n^2}}}\]
bits, achieving the optimal value when $\l = \log n + O(\log\log(U/n))$. However, when $n$ changes, the value of $\l$ has to change, and the data structure has to be reconstructed. But, when we reconstruct the data structure, we no longer have access to the original keys (we can only access the fingerprints and future fingerprints), and thus we lose some of the information going forward. 
This is where our future fingerprint and rebuilding algorithm comes in to play. Note that the rebuilding process is the reason why the midyard is not empty, as when we decrease the fingerprint length, keys that were not colliding before may now collide. Because we cannot simply send these new collisions to backyard (since we do not have the full key), we instead send these keys to the midyard, adding a layer of complexity to our data structure. 

The idea for the rebuild process is simple: The standard fingerprint length may change, so we try to adjust all fingerprints to length $\l$. Specifically, if $\l$ increments by one, for every key $y$ in the frontyard, we try to take the first bit of $y$'s future fingerprint and append that bit to the end of its fingerprint; if $\l$ decrements by one, for every key $y$ in the frontyard with length of previous $\l$, we try to move the last bit of its fingerprint back to its future fingerprint. Note that two different fingerprints may collide after this operation; in this case, we extend the fingerprints of both keys to include all of their future fingerprint bits, and we throw both keys' fingerprints into the midyard. Moreover, since keys may have been deleted since the last rebuild, to make sure most of the keys are stored in the frontyard, we need to lift some fingerprints/keys in the midyard/backyard to the frontyard. \cref{alg:rebuild} explains this procedure in full detail.

\begin{algorithm}[ht]
\caption{Rebuild}\label{alg:rebuild}
\DontPrintSemicolon
    Let $\l, \, \l_f$ be the new standard fingerprint length and the standard future fingerprint length. They differ from the original values by at most 1. \;
    \If{$\l$ was increased by $1$}{
        \tcp{Extend fingerprint lengths from $\l - 1$ to $\l$ if necessary.}
        \For{each fingerprint $\tilde{x}$ such that $\tilde{x}$ has $\l - 1$ bits}{
            \If{$\tilde{x}$'s future fingerprint $x_f$ is nonempty}{
                Move the first bit of the future fingerprint to the fingerprint. \;\label{line:extend_fp}
            }
        }
    }
    \If{$\l$ was decreased by $1$}{
        \tcp{Shrink fingerprint lengths from $\l + 1$ to $\l$ if possible.}
        \For{each fingerprint $\tilde{x}$ in the frontyard that has length $\l + 1$}{
            Let $u$ be $\tilde{x}$'s depth-$\l$ ancestor. \;
            \uIf{$\tilde{x}$ is the only fingerprint in $u$'s subtree in the frontyard and midyard}{
                Move the last bit of $\tilde{x}$ to its future fingerprint.
            }\Else{
                Move $\tilde{x}$ to the midyard; move all bits in $\tilde{x}$'s future fingerprint to its fingerprint.
            }
        }
    }
    \tcp{Lift keys from the midyard to the frontyard if possible.}
    \For{each depth-$\l$ vertex $u$ on the Trie}{
        \If{no fingerprint is an ancestor of $u$ and there is only one key $x$ in the midyard in $u$'s subtree}{
            Move $x$ to the frontyard; take the length-$\l$ prefix of $x$ (which equals $u$) to be its new fingerprint, and let the remaining bits of $x$ in the data structure be its new future fingerprint. \;
        }
    }
    \tcp{Lift keys from the backyard to the frontyard if possible.}
    \For{each key $x$ in the backyard}{
        Let $u$ be the length-$\l$ prefix of $x$. \;
        \If{no fingerprint is an ancestor of $u$ or in the subtree of $u$}{
            Move $x$ to the frontyard; take the length-$\l$ prefix of $x$ (which equals $u$) to be its new fingerprint, and let the remaining bits of $x$ be its new future fingerprint. \; \label{line:lift_backyard}
        }
    }
    \tcp{Truncate future fingerprints that are longer than the standard length $\l_f$.}
    \For{each fingerprint $(\tilde{y}, y_f, \Value_y)$}{
        Truncate $y_f$ to length $\ell_f$. \;\label{line:truncate_ffp}
    }
\end{algorithm}

\smallskip

We remark that only fingerprints with length $\l$ can have nonempty future fingerprints, because otherwise, \cref{line:extend_fp} will extend the length of the fingerprint.
We also point out that the rebuilding process is the only step in our data structure where keys are moved into the midyard; between two rebuilds, the midyard will only receive deletions with no insertions. Similarly, the rebuild is the only place where fingerprints shorter than $\l$ bits can be generated. These facts will be crucial in \cref{sec:algorithm} where we make our data structure time-efficient.

Note that, by design, the algorithm maintains the following invariants:

\begin{invariant}\label{lem:noanc}
    No fingerprint is an ancestor of another fingerprint.
\end{invariant}

\begin{invariant}
    For any key $x$ in the current data structure that is not stored in the backyard, its corresponding fingerprint is the unique ancestor fingerprint of $x$.
\end{invariant}

\paragraph{Space usage.}

Now that the correctness of the algorithm is clear, the only remaining task is to analyze the space usage of the data structure. Specifically, we will show that the space usage of the data structure is $nv + O(n \log\log (U/n))$ bits. The main task in this analysis is to bound the number of keys in the midyard and the backyard. We start by lower bounding the length of the fingerprint for each key in the data structure.

\begin{claim}\label{claim:depthlowerbound}
    For any key $x$, if from the moment that $x$ was inserted to the current moment, the size of the key set has never been smaller than $W$, then the fingerprint of $x$ in the data structure must have length $\ge \min\{\ell, \; \log W + 20\log\log (U/W)\}$ (or $x$ is in the backyard).
\end{claim}

\begin{proof}
    First, we lower bound the total length of $x$'s fingerprint and future fingerprint.
    Note that the only way to reduce this quantity is the truncation step of the rebuilding algorithm (\cref{line:truncate_ffp}), which can only reduce the total length to $\l + \l_f \ge \log W + 20 \log\log (U/W) + 2c$
    bits as the data structure always contains at least $W$ keys when $x$ is present.
    
    This implies a lower bound on the length of $x$'s fingerprint in the current state of the data structure, because only the fingerprints with length $\ge \l$ can have nonempty future fingerprints, and thus, either $x$'s fingerprint has length $\ge \l$, or $x$ has no future fingerprint (so $x$'s fingerprint length is at least $\log W + 20 \log \log (U/W)$). The claim holds in both cases.
\end{proof}

Now, we will use this claim to bound the number of depth $\ell$ vertices which contain a fingerprint as a strict ancestor. This claim will form an instrumental part in the bound on the number of bits dedicated to keys in the midyard and backyard. 

\begin{claim}\label{clm:numFrontyard}
    The number of vertices $u$ with depth $\ell$ for which the data structure contains a fingerprint that is $u$'s \emph{strict} ancestor is $O\left(n/\log^{8}(U/n)\right)$.
\end{claim}

\begin{proof}
    For any $W$, we can partition the $n$ keys in the current key set into two parts according to the relative order that they were inserted: the first $W$ inserted keys and the last $n - W$ inserted keys. Then, for each of the last $n - W$ inserted keys, by the previous \cref{claim:depthlowerbound}, its fingerprint length is at least
    \[\min(\l, \, \log W + 20\log\log(U/W)).\]
    Note that when $W \geq 2^c \cdot n / \log^{10}(U/n)$,
    \begin{align*}
    \log W + 20\log\log (U/W) &\geq \log n-10\log\log (U/n) + 20\log\log ((U/n)\log^{10}(U/n)) + c \\
    &\ge \log n + 10 \log \log (U/n) + c
    \geq \ell.
    \end{align*}
    Thus, only the first $O(n/\log^{10}(U/n))$ inserted keys may have a depth $< \ell$.
    
    Furthermore, if a key is maintained at depth $d$, the number of its descendants with depth $\ell$ is $O(2^{\ell-d})$ (the following $\l-d$ bits are arbitrary). Then the number of vertices with depth $\l$ that have a fingerprint as its ancestor is bounded by the sum of $O(2^{\l - d})$ of the first $O(n/\log^{10}(U/n))$ inserted keys, where $d\geq \log W + 20\log\log (U/W)$ and $W$ is the rank of its insertion time. The number of such vertices is therefore at most
    \[
    \begin{aligned}
    & \phantom{{}={}}O\left(\sum_{W=1}^{O(n/\log^{10}(U/n))}2^{\textcolor{red}{\l}-\log W-\textcolor{blue}{20\log\log (U/W)}}\right)\\
    & \leq O\bk*{\sum_{W=1}^n 2^{\textcolor{red}{\log n + 10\log\log(U/n)} - \log W - \textcolor{blue}{(2\log\log(n/W) + 18\log\log(U/n))}}}\\
    & \leq O\left(\sum_{W=1}^{n}2^{\log (n/W) - 2\log\log (n/W) - 8\log\log (U/n)}\right).
    \end{aligned}\]
    Now, we divide $\{1,2,\dots,n\}$ into $\log n$ intervals by the value of $j=\lceil\log(n/W)\rceil$. There are at most $O(n/2^j)$ integers $W$ in $\{1,2,\dots,n\}$ that have $\lceil\log(n/W)\rceil\geq j$. So after we change it to summing up by $j$, it becomes
    \[
    \begin{aligned}
    & \phantom{{}={}} O\left(\sum_{j=1}^{\log n}\frac{n}{2^j}2^{j - 2\log j - 8\log\log (U/n)}\right)\\
    &= O\bk*{n\sum_{j=1}^{\log n} \bk*{\frac{1}{j^2} \cdot \frac{1}{\log^8(U/n)}}}\\
    & = O\left(\frac{n}{\log^8(U/n)}\sum_{j=1}^{\log n}\frac{1}{j^2}\right)\\
    & = O\left(\frac{n}{\log^{8}(U/n)}\right).
    \end{aligned}
    \]
\end{proof}

The following claim shows a necessary condition for any key $x$ to appear in the midyard or backyard.

\begin{claim}\label{clm:condInBackyard}
    If a key $x$ is currently in the midyard or backyard with $u$ being its length-$\l$ prefix, then there exists another key $y$, such that:
    \begin{enumerate}
        \item either $y$ contains $u$ as a prefix, or $y$'s fingerprint in the data structure is a prefix of $u$;
        \item either $y$ was in the data structure at the last rebuild, or $y$ was inserted after the last rebuild. (Note that we do not require $y$ to be present in the current state of the data structure.)
    \end{enumerate}
\end{claim}

\begin{proof}
    If $x$ was inserted before the last rebuild, then we can find such a $y$ which was present at the moment of the last rebuild, since the last rebuild did not lift $x$ to $u$.
    Otherwise, as $x$ was inserted after the last rebuild, $x$ must be in the backyard instead of the midyard, and the key that blocks it from being inserted into the frontyard gives us the desired $y$.
\end{proof}

We are almost ready to prove that the size of the midyard and the backyard is small. The previous claim says that if a depth-$\l$ node is in the midyard or backyard, then there must have been more than one key in its subtree.
The following claim shows that the number of such nodes is small.

\begin{claim}\label{clm:sizeOfMidyard}
    For $n$ random keys from the universe, letting $n'$ be the number of distinct length-$\l$ prefixes of these keys, we have that \[n - n' \leq O\left(\frac n{\log^8(U/n)}\right)\] with high probability in $n$.
\end{claim}

\begin{proof}
    Define $Y_i$ as the indicator variable of the event that the $i$-th hashed key collides with some $j$-th hashed key where $j<i$ (where a collision means that they have the same length $\l$ prefix). Then, the value of $n - n'$ is exactly the sum of the $Y_i$'s, because $n - n'$ is $n$ minus the number of distinct prefixes (and thus exactly counts the number of collisions). We then define auxiliary variables $Y'_1,Y'_2,\dots,Y'_n$ where $Y'_i = Y_i - \frac{i-1}{2^{\ell}}$. Note that \[
    \forall i, \E[Y'_i\mid Y'_1,\ldots,Y'_{i-1}]\leq 0.
    \]
    By Azuma's Inequality, \[\Pr\left[\sum_{i=1}^n Y'_i\geq O\left(\frac{n}{\log^8 U}\right)\right]\leq \exp\left(-\Omega\left(\frac{n}{\log^{16} U}\right)\right)=U^{-\omega(1)}.\] Furthermore,
    \[\sum_{i=1}^n \frac{i-1}{2^\l}\leq O\left(\frac{n}{\log^8 (U)}\right).\]
    which means that $\sum_{i = 1}^n Y_i = n - n'$ will not exceed $O(n/\log^8(U))$ with high probability. Now, because $\log(U/n) \leq \log(U)$, we see that $\frac{n}{\log^8 (U/n)} \geq \frac{n}{\log^8(U)}$, and so this sum is also bounded by $O\left ( \frac{n}{\log^8 (U/n)}\right )$, as we desire. 
\end{proof}

Now we are ready to give the following key claim, which bounds the number of fingerprints in the midyard and the backyard:

\begin{claim}\label{clm:numBackyard}
    At any given moment, the number of fingerprints with length $>\ell$ (in the midyard) and the number of keys in the backyard is $O(n/\log^8(U/n))$ with high probability in $n$.
\end{claim}
\begin{proof}
First, we make an observation that will be useful in the following proof. If we consider our data structure at any point in time between consecutive rebuilds, then we only have \emph{deletions} of keys of length $< \ell$, as any keys that are inserted are either at depth $\ell$ or in the backyard. We call this property being ``decremental'', as the number of these keys is not increasing. 

Now, the remainder of the proof will proceed by looking at each depth $\ell$ node $u$ in our Trie, and bounding the number of descendants that this node has. So, let us fix a depth-$\l$ node $u$ and analyze the number of keys in the subtree of $u$, excluding $u$ itself. Essentially, we are counting all the keys $x_{u, 1}, \dots x_{u, s_u}$, where each $x_{u,i}$ is a key that contains $u$ as a prefix. Note that each $x_{u,i}$ either existed at the moment of the last rebuild or was inserted after the last rebuild. Because keys can be inserted and deleted repeatedly, for analysis, we only include a key $x_{u,i}$ in our count at most a single time.

Next, note that if $s_u = 1$, and $x_{u,1}$ is in the frontyard, then we can skip this case, as there will be no fingerprints in $u$'s subtree that are in the backyard or midyard. Thus, we focus specifically on the cases where $s_u = 1$, and $x_{u,1}$ is not in the frontyard, and when $s_u > 1$.
    
Now, using the above logic, we can create a bound on the number of keys in the midyard and the backyard with $u$ as a prefix as at most:
    \[\sum_{\text{depth-}\l \text{ node }u} 2(s_u-1) + \sum_{\text{depth-}\l \text{ node }u} \mathbf{1}[s_u=1]\cdot \mathbf{1}[x_{u,1} \text{ is not in the frontyard}].\]
    This is because $s_u$ counts all keys (including those in the midyard and backyard) with $u$ as a prefix. If $s_u \geq 2$, then $s_u \leq 2 \cdot (s_u -1)$. Otherwise when $s_u = 1$, we first explicitly check that $x_{u,1}$ is not in the frontyard. The second term encodes this logic and provides an upper bound in this case.
    
    First, we prove the latter term is small. A depth-$\l$ vertex $u$ is counted if and only if $s_u=1$ and there is another key $y$ that satisfies the condition stated in \cref{clm:condInBackyard} with respect to $x_{u,1}$ (i.e., which forces $x_{u,1}$ into the backyard). Note that here, $x_{u, 1}$ and $y$ are both keys, and so must have the same length. Since $s_u=1$, it must then be the case that $y \neq x_{u,1}$, and so $y$ must have a fingerprint which is a strict prefix of $u$. 
    Then, since fingerprints with length $<\l$ are decremental between rebuilds, 
    $y$ must have been in the data structure at the moment of the last rebuild.
    By \cref{clm:numFrontyard}, at the moment of the last rebuild, the number of vertices with depth $\l$ that have a fingerprint as its ancestor is at most $O\left(n/\log^8(U/n)\right)$. Thus the number of possible $u$'s which can find such a $y$ is at most $O\bk*{n/\log^8(U/n)}$, and so the latter term satisfies our desired bound. 
    
    Now, we only need to prove that the former term does not exceed $O\left(n/\log^8(U/n)\right)$. For this, recall that $s_u$ refers to the number of keys that contain $u$ as a prefix, so $\sum_{u: s_u \neq 0} s_u = n$, as every key has \emph{some} $u$ as a prefix. At the same time, when we consider $\sum_{u: s_u \neq 0} -1$, this becomes exactly the (negative) \emph{number} of distinct prefixes for our keys. In the notation of \cref{clm:sizeOfMidyard}, we refer to this quantity as $-n'$. To conclude then, we see that 
    $\sum_{u: s_u \neq 0} (s_u  -1) = n - n'$, which we can then bound by $O\left(\frac{n}{\log^8(U / n)}\right)$ using \cref{clm:sizeOfMidyard}.

    By combining our bounds for the first and second term in the above expression, we obtain our desired bound. 
\end{proof}

The following lemma will conclude the proof of the total space consumption.

\begin{lemma}
    The space consumption is $nv+O(n\log\log(U/n))$ with high probability in $n$.
\end{lemma}

\begin{proof}
    We calculate the space usage of the frontyard, midyard, and backyard separately. 
    
    To start, we bound the space used in the frontyard. Assume that the number of fingerprints in the frontyard is $m$. For each of them, we need to store a value of $v$ bits and a future fingerprint of $O(\log \log (U/n))$ bits. Since there are $2^\l\cdot 2$ possible fingerprints of length $\leq \l$, the total space used in the front yard is \[
    \begin{aligned}
        mv + \log \binom{2^{\ell} \cdot 2}{m} + O(m\log\log (U/n)) &\leq mv + n\log (2^{\ell + 1} / n) + O(n \log\log (U/n))\\
        &= mv + O(n\log\log(U/n))
    \end{aligned}\]
    bits, where the equality holds because $2^\ell = O(n\log^{10}(U/n))$.

    Next we bound the size of the midyard and backyard, where there are $n - m$ elements. Observe that here, keys have no future fingerprints. Additionally, the universe size for keys in the midyard and backyard together is $2U$, since the keys are prefixes of entire keys (which come from a universe of size $U$). Thus, the space usage of these keys is
    \[
    \begin{aligned}
        (n-m)v + \log \binom{2U}{n-m}&\leq (n-m)v + O\left(\frac{n}{\log^8(U/n)}\cdot \log\left(\frac{U \log^8 (U/n)}{n}\right)\right) \\
        &= (n-m)v + O(n)
    \end{aligned}\]
    with high probability in $n$, using the fact that $n - m\leq O(n/\log^8(U/n))$ with high probability according to \cref{clm:numBackyard}.
\end{proof}

\section{Improving the Time Complexity Bounds}
\label{sec:algorithm}

A fundamental property of our algorithm is that, between two rebuilds, we either insert keys to vertices on the Trie with depth $\l$ or into the backyard. So, it is easy for the insertion operations to be made time-efficient.
The remaining challenge is in bounding the time complexity of the \emph{queries}, and more specifically, in finding the ancestor fingerprint of a given leaf node. This will be the primary focus of the current section.

Recall that from the previous section, the Trie is split into two parts---the frontyard and the midyard. Since no fingerprint is an ancestor of another fingerprint, the procedure of finding the ancestor fingerprint of a leaf $x$ can be done by searching for the ancestor fingerprint separately in the frontyard and the midyard.
In this sense, the frontyard and midyard can be implemented and analyzed independently. Finally, as in the previous section, all distinct inserted/queried keys will be regarded as uniformly random from the universe $[U]$. We will show how to overcome this assumption in \cref{sec:hash_function}.

\subsection{Frontyard}

\paragraph{Answering Queries} 
To start, we introduce the data structures that we will require for answering queries. As motivation, suppose that we are queried with a depth-$\ell$ vertex $u$, but that $u$ itself is not a fingerprint. This means that $u$ is potentially the descendant of a fingerprint of length $< \ell$, and it becomes imperative for us to actually locate this ancestor fingerprint. To do this, we adopt \emph{predecessor data structures}. We associate vertices at depth $\l$ in the Trie with numbers in $[2^\l]$, and for any fingerprint at depth $\leq \ell$, we say that it ``covers'' the interval of these numbers that corresponds with its depth $\ell$ descendants. Now, given a depth $\ell$ vertex, our goal will then be to \emph{find} its ancestor fingerprint if it exists (i.e., the query will return the fingerprint $\tilde{x}$ if and only if the binary number corresponding to the given vertex $u$ is in $\tilde{x}$'s interval). Recall that these fingerprints themselves have no ancestor-descendant relationships (as per \cref{lem:noanc}), and as a result, these intervals induced by fingerprints do not overlap. To support this type of operation, we construct two predecessor data structures: the first we call $\mathrm{Predecessor}$, which stores all the left endpoints of these intervals, and the other, $\mathrm{Successor}$, which stores all the right endpoints. Given a query vertex $u \in [2^\l]$, we can then find the largest left endpoint $L$ with $L \leq u$ and the smallest right endpoint $R$ with $R \geq u$ by querying $\mathrm{Predecessor}$ and $\mathrm{Successor}$ respectively. Note that, even though we refer to the data structures as ``Predecessor'' and ``Successor'' respectively, both will be instances of so-called \emph{predecesor data structures}.

In addition, once we have found a fingerprint, we need to be able to check (1) what its corresponding value and future fingerprints are, and (2) whether the fingerprint has been deleted. We store the information required for (1) in a dynamically resizable hash table we call $\mathrm{ValueHash}$, and we store (2) in a hash table we call $\mathrm{DeletedHash}$. In summary then, these are the $4$ data structures we store in the frontyard:

\begin{enumerate}
    \item A dynamically resizable hash table that maps fingerprints to values and future fingerprints ($m$ keys). We refer to this table as $\mathrm{ValueHash}$.
    \item A hash table that stores whether a fingerprint with depth $< \l$ has been deleted (at most $m_0$ keys). We refer to this table as $\mathrm{DeletedHash}$.
    \item Two static predecessor data structures of at most $m_0$ keys with universe $[2^\l]$, where $2^{\l} = O(n\log^{10}(U/n))$. We call these data structures $\mathrm{Predecessor}$ and $\mathrm{Successor}$.
\end{enumerate}

Now, we will explain how to use these data structures to answer queries. As before, suppose that we are queried with a depth-$\l$ vertex $u$, and we wish to find its ancestor fingerprint. To start, we can check if $u$ itself is a fingerprint by querying the $\mathrm{ValueHash}$ table. If it is, we can also find its value, and thus respond to the query. Otherwise, the ancestor fingerprint of $u$ must be of depth $< \ell$. So, we find the largest left endpoint $L$ with $L \leq u$ and the smallest right endpoint $R$ with $R \geq u$ by querying $\mathrm{Predecessor}$ and $\mathrm{Successor}$ respectively. If $u$ has an ancestor fingerprint $\tilde{x}$ with depth $< \l$, then $\tilde{x}$ must be the lowest common ancestor (LCA) of $L$ and $R$ (which can be calculated in $O(1)$ time using bit operations). We then check $\mathrm{DeletedHash}$ and $\mathrm{ValueHash}$ to determine if it is a valid fingerprint stored in the data structure.

We implement this logic algorithmically below:

\begin{algorithm}[H]
    \caption{$\mathrm{QueryInformation}(u)$}
    \If{$u \in \mathrm{ValueHash}$}{
    \Return{$\mathrm{ValueHash}(u)$}
    }
    Let $L \leftarrow \mathrm{Predecessor}(u)$. \\
    Let $R \leftarrow \mathrm{Successor}(u)$. \\
    Let $\tilde{x}$ be the LCA of $L, R$. \\
    \If{$\tilde{x} \in \mathrm{DeletedHash}$}{
    \Return{Deleted.}
    }
    \Else{
    \Return{$\mathrm{ValueHash}(\tilde{x})$.}
    }
    
\end{algorithm}



Importantly, note that between consecutive rebuilds, the set of fingerprints with depth $< \l$ can shrink (because of deletions) but cannot grow (because new insertions do not insert such fingerprints). Thus, if the corresponding fingerprint for $u$ was $\widetilde{x}$, but now $\widetilde{x}$ has been deleted, then we can be sure that $u$ has \emph{no} fingerprint of length $< \ell$. This is because, as mentioned before, \cref{lem:noanc} guarantees that the induced intervals of covered keys by fingerprints in the Trie are disjoint, and once one such interval is removed, no new interval is added to replace it. 

Additionally, this so-called ``decremental'' property (i.e., that fingerprints of length $< \ell$ are only deleted) allows for $\mathrm{DeletedHash}$ to monitor exactly which fingerprints have been deleted. Whenever such a fingerprint is deleted, $\mathrm{DeletedHash}$ is simply updated to keep track of the new state.


Having established the correctness of this algorithm, we devote the rest of this section to analyzing the space and time complexity of the above algorithm. We assume that the number of fingerprints with depth $<\l$ is $m_0$ and the number of fingerprints in the frontyard is $m\leq n$. The following fact, which intuitively says that most fingerprints in the frontyard lie at depth $\l$, will play a key role in our later analysis:

\begin{claim}\label{clm:numPseudo}
    The number of fingerprints that have depth $<\l$ is $m_0\leq O(n/\log^{8}(U/n))$.
\end{claim}

\begin{proof}
    The number of fingerprints with depth $<\l$ is smaller than the number of vertices with depth $\l$ that have a fingerprint as an ancestor, which is $O(n/\log^8(U/n))$ according to \cref{clm:numFrontyard}.
\end{proof}

With this essential claim, we are able to create a better bound on the space complexity of the data structure.

\paragraph{Predecessor data structures.} In the previous section, we made use of two predecessor data structures, namely $\mathrm{Predecessor}$ and $\mathrm{Successor}$. 
Towards our ultimate goal of bounding the overall space and time complexity of our retrieval data structure, we dedicate this section to bounding the space and time complexities of these predecessor data structures. To start, we break the analysis into two cases.

\textbf{Case 1:} If $n > \sqrt U$, we construct our predecessor data structures using the following theorem:

\begin{theorem}\label{thm:staticCase1}
    Given a set $S$ of $n$ integers in universe $[N]$, there is a static predecessor data structure that supports predecessor queries in $O(1)$ time, uses \[\log \binom{N}{n} + O\left(\frac{N}{\log^{10} N}\right)\] bits of space, and can be constructed in $O\bk*{\log \binom{N}{n}+\frac{N}{\log^{10} N}}$ time.
\end{theorem}
\begin{proof}
    We construct a $\{0,1\}$-valued-sequence $a_0, a_1, \dots, a_{N-1}$ of length $N$ where $a_x=1$ if and only if $x\in S$. Next, we construct a succinct rank/select data structure over this sequence to support two types of queries: 
    \begin{enumerate}
        \item Finding the position of the $k$-th 1 in the sequence.
        \item Counting the number of $1$'s before a given index $i$.
    \end{enumerate} 
    As per \cite{4690964}, this rank / select data structure can be built within \[\log \binom{N}{n} + O\left(\frac N{\log^{10}N}\right)\] bits of space and supports queries in constant time. 
    
    Finally,  it remains to show how we can use this data structure to build our static predecessor data structure. The key observation is that given a predecessor query with integer $y \in [N]$, we can report the correct answer by first performing a rank query for the number of $1$'s before index $y$ (which is exactly $t=a_0+a_1+\cdots + a_{y-1}$) and then performing a select query for the index of $t$-th $1$ in the sequence. This index of the $t$-th $1$ in the sequence will then exactly reveal the location of the final $1$-index before $y$, as we desire.
\end{proof}

Now, if we instantiate \cref{thm:staticCase1} in our setting, the universe size is $2^\l = \Theta(n \log^{10}(U/n))$ and the number of keys is $O(n/\log^8(U/n))$. This means the total space consumption (in bits) used by the predecessor data structure is
\[\log \binom{2^\l}{n/\log^8(U/n)} + O\bk*{\frac {2^\l}{\l^{10}}} = O\bk*{\frac{n\log\log(U/n)}{\log^8(U/n)} + \frac{n\log^{10}(U/n)}{\log^{10}n}} = O(n),\]
where the last equality holds due to our assumption $n \ge \sqrt U$.

\textbf{Case 2:} If $n \leq \sqrt U$, we use a different approach that relies on fusion trees.


\begin{theorem}[Fusion Tree, \cite{fusion}]\label{thm:fusion}
    For a key universe $[U]$, word size $w\geq \log U$, and $n$ keys, there exists a \emph{fusion tree} data structure that stores the set of keys in $O(nw)$ space and supports predecessor queries in time $O(\log_w n)$. The fusion tree can be constructed in time $O(n)$ and space $O(nw)$ from the sorted keys.
\end{theorem}

As shown in \cite{fusion}, (and from the above claim), fusion trees can be implemented in space $O(n\log U)$ by restricting the word size $w$ to be $O(\log U)$.

Now, to implement our predecessor data structure, we divide the universe of size $2^\l = \Theta(n \log^{10}(U/n))$ into $n/{\log^2(U/n)}$ equal-sized blocks, where each block has size $\Theta(\log^{12}(U/n))$. For each block, we then maintain a static fusion tree
with word size $O(\log U)$, which stores all numbers in this block. Given the queried integer, we first find the block containing it and then query the corresponding fusion trees for this block. If we find a valid key in this block, we are done. Otherwise, there is no such key within this block, so the answer to the predecessor query is the maximum key occurring before the block. Since the number of blocks is only $n/\log^2(U/n)$, we can directly store the answers for each possible block in this case.

Given a query, we also need to be able to find the corresponding fusion tree in storage, so we store a pointer to the fusion tree, which costs $O(\log U)$ bits of space for each block. Thus, the total space consumption is \[O\bk*{\underbrace{\frac{n\log U}{\log^8(U/n)}}_{\text{fusion trees}} + \underbrace{\frac{n\log U}{\log^2(U/n)}}_\text{directly stored answers} + \underbrace{\frac{n\log U}{\log^2(U/n)}}_{\text{pointers}} } = O(n).\] The time complexity is $O(1)$ since the size of a block is at most $\log^{12}(U/n)$, which is at most a polynomial of the word size.




\paragraph{Hash tables.} Now that we have bounded the space and time complexities of our predecessor data structures, we turn our attention to our hash tables. As a starting point, we use the following theorem:

\begin{theorem}[Generalization of Theorem 5 of \cite{10.1145/3519935.3519969}]\label{thm:dictionarySpace}
    Let the key universe be $[U]$, the length of values be $v \leq O(\log U)$ bits and the number of key-value pairs $n\geq \log^{20}U$. Then for any constant $k>0$, there is a dynamically resizable dictionary that stores a set of key-value pairs, which supports insertions/deletions in time $O(k)$ and queries in time $O(1)$, and offers the following space guarantee: If the current number of keys is $n$, then the total space usage is \[\log \binom{U}{n} + nv + O\bk*{n\log^{(k)}n}\] bits. The running time and space guarantees hold with high probability in $n$. Besides standard dictionary operations, this dictionary supports the fetching of an arbitrary key-value pair stored in the dictionary in time $O(1)$, so as to enable the listing of all keys in the data structure in time $O(n)$.
\end{theorem}

\begin{proof}[Proof sketch]
Essentially, the dictionary in \cite{10.1145/3519935.3519969} breaks the $n$ key value pairs into $O(n/\log n)$ parts. Each part contains a fixed number of key-value pairs, except for a single, special part we designate.

The above theorem relies on two generalizations beyond the work of \cite{10.1145/3519935.3519969}. The first is to use value sizes $v\leq O(\log U)=n^{O(1)}$ instead of $v\leq n^{o(1)}$. As mentioned in \cite{10.1145/3519935.3519969}, in each part, the values are stored in a specific array, and access to this array is guaranteed because each part stores a pointer to its array. Since the total space usage is still $n^{O(1)}$ bits, we can bound the length of each pointer by $O(\log n)$ bits, which, in conjunction with the $O(n/\log n)$ parts, guarantees that the total number of bits allocated to the pointers is $O(n)$ bits. 

The second generalization is that the dictionary can support fetching an arbitrary key-value pair stored in it. This operation is already included in the implementation of the dictionary's algorithm, as the dictionary uses this operation to ensure that the number of key-value pairs in each part remains consistent as the data structure undergoes insertions and deletions.
\end{proof}

By implementing this theorem with the size of universe being $2^\l$, the number of keys being $m$, the values being $v+O(\log\log(U/n))$ bits, and letting $k$ be a parameter as in \cref{thm:dictionarySpace}, this implies that our dynamically resizable hash table (which we call $\mathrm{ValueHash}$) can be stored in \[
\begin{aligned}
    &\phantom{{}={}} \left(\log \binom{2^\l}{n}+m(v + O(\log\log(U/n))) + O\bk*{n\log^{(k)}n + 2^{\l\epsilon}}\right)\\
    &= mv + O(n\log\log(U/n) + n\log^{(k)}n)
\end{aligned}
\] bits of space. Our other hash table (calleed $\mathrm{DeletedHash}$) 
stores only $m_0=O(n/\log^8(U/n))$ keys, each coming from the same universe of size $2^{\ell}$, and each with value length $O(1)$. Thus, its space complexity is bounded by only $O(n)$ bits of space.

Combined with our space and time characterizations of the predecessor data structure, this yields the following claim:

\begin{claim}\label{clm:frontyardSpace}
    All operations (insert, delete, query) in the frontyard part can be done in $O(k)$ time, and the total space consumption in the frontyard is
\[
\begin{aligned}
    mv + O(n\log\log(U/n) + n\log^{(k)}n)
\end{aligned}
\]
bits,
where $m$ is the number of fingerprints in the frontyard.
\end{claim}

\subsection{Midyard}
\label{sec:midyard}

\paragraph{Answering queries.} Next, we discuss how we find the ancestor fingerprint of any query key when the fingerprint is in the \emph{midyard}. Since fingerprints in the midyard have length $>\l$, we divide them into groups by their length-$\l$ prefix. Taking the view of the Trie, each group represents a nonempty subtree rooted at a depth-$\l$ vertex. For each nonempty subtree in the midyard, we perform the same high-level construction as in the frontyard: For each such depth-$\l$ vertex $u$, we regard all leaves in the subtree of $u$ as binary numbers in $[U/2^\l]$, so that each fingerprint covers an interval; then, we construct two predecessor data structures storing the intervals' left and right endpoints, which allows the queries to be efficiently answered by predecessor searches. As we did in the second case in the frontyard case, we implement the predecessor data structures using fusion trees. We denote these data structures by $\mathrm{PredecessorMidyard}$ and $\mathrm{SuccessorMidyard}$ (note that each of these is really a collection of fusion trees, as each subtree gets its own fusion tree). In the following section, we show that these fusion trees use limited space and have $O(1)$ query time, thus giving us our desired characteristics. 

Note that in addition to supporting these predecessor queries, we also need a way to recover the actual values once we find the ancestor fingerprint. For this, we use a hash table, which we call $\mathrm{ValueHashMidyard}$. Lastly, given a queried key, we also need to be able to \emph{find} the corresponding predecessor data structures for the sub-tree that contains this key. We denote this data structure by $\mathrm{FindPredecessorMidyard}$.
In summary, these will be the three types of data structures we use for the midyard:

\begin{enumerate}
    \item $\mathrm{PredecessorMidyard}$ and $\mathrm{SuccessorMidyard}$, which are both implemented with fusion trees. Note that the trees themselves are implemented on a \emph{per subtree} basis, though we often refer to them collectively.
    \item $\mathrm{ValueHashMidyard}$, which is a hash table that recovers the values once the ancestor fingerprint is calculated.
    \item $\mathrm{FindPredecessorMidyard}$ which is a data structure that maps a key to its corresponding predecessor data structures (since there are different fusion trees for different keys, we need to be able to map them accordingly).  
\end{enumerate}

This logic is implemented in the below algorithm:

\begin{algorithm}[H]
    \caption{$\mathrm{QueryMidyard}(u)$}
    Let $\mathrm{PredecessorMidyard}, \mathrm{SuccessorMidyard} \leftarrow \mathrm{FindPredecessorMidyard}(u)$. \\
    $L \leftarrow \mathrm{PredecessorMidyard}(u)$. \\
    $R \leftarrow \mathrm{SuccessorMidyard}(u)$. \\
    Let $\tilde{x}$ be the LCA of $L, R$. \\
    \If{$\tilde{x} \in \mathrm{ValueHashMidyard}$}{
    \Return{$\mathrm{ValueHashMidyard}(\tilde{x})$.}
    }
\end{algorithm}

In the following paragraphs, we show how to implement these data structures without violating our space constraints. 

\paragraph{Storing fusion trees.} 


Recall from \cref{thm:fusion} that a fusion tree can store a set of $s$ keys, where each key is a $w'$-bit integer, using $O(sw')$ bits of space while supporting predecessor operations in $O(\log_{w'} s)$ time. Here, the word size $w'$ that we use for the fusion tree does not necessarily match the machine word size $w$ for the RAM model; we can choose a smaller word size $w' < w$ that leads to a better space bound and worse operational time, provided $w'$ is still larger than the length of each key. In our application, the key is an integer in $[U/2^\l]$, so its length is at most $\log (U/n)$ bits, and we will use a two-stage strategy to assign word sizes and allocate memory for the fusion trees. 

Before beginning our case analysis, we will require the following claim which bounds the number of fingerprints that can be in any single subtree:

\begin{claim}\label{clm:slotBoundLog2U}
    The number of fingerprints in one subtree is $O(\log^2 U)$ with high probability in $U$.
\end{claim}

\begin{proof}
    When a fingerprint is in the subtree, its original key must also be in this subtree since the fingerprint is a prefix of the key. Let $X_i$ be the indicator variable of the event that the $i$-th hashed key is in this subtree. Then \[\E[X_i] = 2^{-\ell} = O\left(\frac{1}{n\log^{10}(U/n)}\right).\] If the number of fingerprints in the subtree is at least $O(\log^{2} U)$, we must have $X_1 + \cdots + X_n \ge O(\log^{2}U)$. By the Chernoff bound, this occurs with probability at most $\exp(-\Omega(\log^2 U)) = U^{-\omega(1)}$.
\end{proof}

With this, we can present our cases:
\begin{itemize}
\item \textbf{Case 1:} If a Trie subtree contains at most $O(\log^5 (U/n))$ fingerprints, we choose the word size $w' = O(\log (U/n))$, so that the fusion tree uses at most $O(\log^6 (U/n))$ bits of space, and can perform operations in $O(1)$ time. We reserve a (fixed-size) chunk of memory of $\Theta(\log^6 (U/n))$ bits for it. We call these fusion trees of the first type. 
\item \textbf{Case 2:} If a Trie subtree contains more than $\Omega(\log^5 (U/n))$ fingerprints, we choose the word size $w' = \min\{\log U, \, U/n\}$. In this case, the fusion tree still takes constant time for each operation: If $w' = \log U$, we take advantage of \cref{clm:slotBoundLog2U}, which shows that the number of fingerprints in each subtree is at most $O(\log^2 U)$ with high probability, which in turn implies a constant operational time via \cref{thm:fusion}; if $w'=U/n$, we know that the number of keys in each subtree is bounded by $U/n$, so the operations are also constant-time.
Since the number of keys is bounded by $U/n$ and the word size is $w' \leq U/n$, the number of bits used by a fusion tree is at most $O((U/n)^2)$ bits. So, we allocate a chunk of memory of $O((U/n)^2)$ bits for the fusion tree.
We call these fusion trees of the second type. 
\end{itemize}

Now, our goal is to bound the number of bits we allocate for the fusion trees. Observe that by \cref{clm:numBackyard}, the number of fingerprints in the midyard or backyard is $O(n/\log^8(U/n))$, so the number of fusion trees of the first (smaller) type does not exceed $O(n/\log^{8}(U/n))$. The following claim bounds the maximum number of fusion trees of the second (larger) type.

\begin{restatable}{claim}{boundLargeBlock}\label{clm:boundLargeBlock}
    Given $n$ random keys, the number of length-$\l$ strings $u$ such that there are at least $\Omega(\log^5(U/n))$ keys containing $u$ as a prefix is $O(n^3/U^2)$ with high probability. Moreover, when $U\geq n^{1.01}$, no such $u$ exists.
\end{restatable}

The claim can be proved by careful calculation based on concentration bounds, and we defer the proof to \cref{sec:hash_function}. This claim implies that the number of fusion trees of the second (larger) type is bounded by $O(n^3 / U^2)$ with high probability, since every fingerprint in $u$'s subtree should have $u$ as its length-$\l$ prefix.
Based on \cref{clm:slotBoundLog2U,clm:boundLargeBlock},
the total space usage for the fusion trees (of both the first and second types) is
\[O\bk*{\frac n{\log^8(U/n)}\cdot \log^5(U/n) + \frac{n^3}{U^2}\cdot (U/n)^2}=O(n)\]
bits, as desired.

\paragraph{Finding Fusion Trees}
In addition to efficiently performing predecessor queries on fusion trees, we also need to be able to \emph{find} the corresponding tree for each queried key in $O(1)$ time. Thus, we need a data structure that maps a subtree to its corresponding chunk in memory.
Recall that no insertions are made into the midyard between rebuilds, so the fusion trees only need to handle queries and deletions, and we can ensure that the sizes of the fusion trees will never increase. This means that we do not have to create new chunks to store the fusion trees.
Thus, for our purposes, it suffices to construct a static injective function that maps from an index of the subtree to an index of a chunk. We use the following result on perfect hashing to achieve this:


\begin{theorem}[Theorem 3 of \cite{demaine2006dictionariis}]
     For any integers $n', t \in \mathbb{Z}^+$, and a universe size $U\geq 2n$, there is a data structure that stores $n'$ keys and maintains an injective function mapping from the keys to $[n'+t]$ that supports constant time query, using $O(n' \log\log (U/n')+ n' \log \frac{n'}{t+1} )$ bits of space.
\end{theorem}


The keys in the data structure correspond to the indices of the existing subtrees, with the number of keys being \( n' = O(n/{\log^8(U/n)}) \) in one type or \( n' = O(n^3/U^2) \) in the other. In our instantiation of the data structure, we also define \( t \) to be equal to the number of keys, i.e., $t$ is also \( O(n/{\log^8(U/n)}) \) in one type or \( O(n^3/U^2) \) in the other. As a result, we can map the index of a subtree to an index in $[n'+t]$ in constant time, with $n'+t$ smaller than \( O(n/{\log^8(U/n)}) \) or \( O(n^3/U^2) \), depending on the type, and store only the corresponding number of chunks. Thus, the overall space complexity of this injective function is 
\[
O\left (n' \log\log (U/n')+ n' \log \frac{n'}{t+1} \right) = O\left ( \frac{n}{\log^8(U/n)} \log\log\left(\frac{U \log^8(U/n)}{n}\right) + \frac{n}{\log^8(U/n)} \right) = O(n),
\]
in one case, and in the other case the space complexity is
\[
O\left (n' \log\log (U/n')+ n' \log \frac{n'}{t+1} \right) =O\left(\frac{n^3}{U^2} \log\log\frac{U^3}{n^3} + \frac{n^3}{U^2}\right) = O(n).
\]

Finally, we come to the end of our analysis.

\begin{claim}\label{clm:midyardSpace}
    For $k$ a parameter of our choosing, the space consumption of the midyard is $m'v + O(n\log^{(k)} n)$ with high probability, where $m'$ is the number of fingerprints in the midyard and all operations can be done in $O(k)$ time.
\end{claim}

\begin{proof}
    Using \cref{thm:dictionarySpace}, the space consumption of the value-storing hash table is \[
    \left(m'v + \log \binom{U}{m'} + O(n\log^{(k)}n)\right) =m'v + O(n\log^{(k)}n).\]
    As stated above, the space consumption of $\mathrm{FindPredecessorMidyard}$, $\mathrm{PredecessorMidyard}$ and \\$\mathrm{SuccessorMidyard}$ is $O(n)$ with high probability. Thus the total space consumption is \[m'v + O(n\log^{(k)} n),\]
    as we desire.
\end{proof}

\subsection{Backyard}
\label{sec:backyard}

Next, we show how to improve the running time bounds for keys in the \emph{backyard}. Fortunately, handling this case is much simpler. Indeed, by \cref{clm:numBackyard}, the backyard contains a set of at most $O(n/\log^8(U/n))$ keys. Thus, we can simply invoke \cref{thm:dictionarySpace} for this key set of size $O(n/\log^8(U/n))$, and with $k = O(1)$, thereby just storing the entire backyard in a simple resizable hash table. We call this hash table $\mathrm{ValueHashBackyard}$.

This yields the following claim:




\begin{claim}\label{clm:backyardSpace}
    All operations in the backyard part can be done in $O(1)$ time, and the total space consumption in the backyard is $O(n)$ bits.
\end{claim}

\subsection{Rebuilds}
\label{sec:rebuild}

Together, the preceding sections show that \emph{between} rebuilds, the operations can be supported both time and space efficiently. In this section, we explain how to implement the rebuilding phase itself in a time and space efficient manner. Recall that the rebuild should satisfy several properties: First, it should recover all fingerprints (using the pseudocode in \cref{alg:rebuild}). Then it should rebuild all data structures used in the frontyard, midyard, and backyard. 

We first analyze the time complexity of rebuilds. Enumerating all the keys in the hash tables will cost $O(n)$ time, and changing them will also run in $O(n)$ time with high probability. Reconstructing all the data structures also costs $O(n)$ time because we are only re-inserting $O(n)$ keys. 

It is worth mentioning that in order to avoid duplicating the $nv$ term in the space usage, we need to be very careful with the values of the keys. Specifically, we need to move the values of the keys one by one, from the old hash table to the new one. Recall that we call rebuild whenever $n/100$ updates have passed since the last rebuild. Thus the rebuild operation costs amortized $O(1)$ time. 

These simple modifications give us a data structure for the resizable retrieval problem, costing amortized $O(1)$ time per operation and \[nv + O(n\log\log(U/n) + n\log^{(k)}n ) \] bits of space with high probability. Note that the run-time is only constant under amortization, as this rebuild phase necessarily costs $O(n)$ time. In the next paragraph, we show how we can avoid this need to amortize.

\paragraph{Deamortizing}

First, we can observe that the only barrier to removing the amortization in the running time is the rebuilding phase, as between rebuilds our operations are supported in constant time (with high probability). To achieve $O(1)$ update time with high probability, our only modification will be to split the process of rebuilding the data structure into $O(n)$ steps, where we do $O(1)$ calculations along with each step. This then makes the update time complexity $O(1)$ with high probability.

However, this is hard to implement if we have only one copy of the data structure: indeed, while we would be performing our so-called \emph{deferred} rebuild, we would still be receiving incoming operations, and would have to simultaneously support operations with different key lengths. Thus, a simpler approach is just to create a \emph{copy} of the data structure. Now, while operations are being received, we query our copy of the data structure, while simultaneously performing our rebuild on the original data structure. Unfortunately, this simple solution is too space intensive, as we end up storing our key-value pairs each \emph{twice}. Thus, our final improvement is to introduce a bucketing scheme, which allows for the rebuild to be implemented on a per-bucket basis, therefore requiring less space.


Formally, let us define $P$ as the minimum power of $2$ greater than $\log^2 U$. We want to split the data structure into $P$ different parts and rebuild them one by one.
Thus, for any $x \in U$, we regard the first $\log P$ bits of $x$ as the label of the part that it belongs to. Now, within each part, we view it as a subproblem with universe $[U/P]$ (where the labels in this smaller universe are given by the remaining bits of $x$). 

Now, let us assume that the number of keys in each part is $n_0,n_1,\dots,n_{P-1}$. The following claim bounds the range of each $n_i$.

\begin{claim}\label{clm:partnotbig}
    For any $i\in [P]$, $n_i = \Theta\bk*{n/P}$ with high probability.
\end{claim}

\begin{proof}
    Let $X_j$ be the indicator variable of the event that the $j$-th key belongs to part $i$. Then $\E[X_j] = 1/P$. Thus the number of keys in part $i$ is $n_i = X_1 + X_2 + \cdots + X_n$, and the expected number of keys is $\mathbb{E}[n_i] = \mathbb{E}[ X_1 + X_2 + \cdots + X_n] = n/P$. Since $n/P\geq \log^2 U$, by Chernoff bound, $n_i=\Theta(n/P)$ occurs with high probability in $n$.
\end{proof}

Now, by using \cref{clm:frontyardSpace}, \cref{clm:midyardSpace}, and \cref{clm:backyardSpace} (with the parameter $n$ instead set to be $n / P$), we see that the total space complexity of the data structure when broken up into $P$ parts is given by \[
\begin{aligned}
    &\phantom{{}={}}\sum_{i\in [P]} \bk*{n_iv + O\bk*{n_i\log\log\bk*{\frac{U}{n_i P}} + n_i\log^{(k)}(n_i)}}\\
    &= nv + O\bk*{P\cdot \frac nP \log\log\bk*{\frac {U}{(n/P)\cdot P}} + n\log^{(k)}n}\\
    &= nv + O(n\log\log(U/n) + n\log^{(k)}n).
\end{aligned}\]
Note that this holds with high probability in $\Theta(n/P)$, and since we have assumed that $n\geq \log^{20}U$, therefore it is also with high probability in $n$.

For simplicity and also to avoid some potential problems when we remove our assumption of the keys being uniformly random (see \cref{sec:hash_function} for a discussion of this matter), we do not split the backyard into $P$ parts, and instead, we store a single, unified backyard which contains complete keys along with their associated values. Moreover, note that in each part, the value of $\l,\l_f$ still changes by at most $1$ between rebuilds with high probability, since during the $n$ operations, exactly $\Theta(n/P)$ operations belong to the part. All in all, there are $P$ frontyards, $P$ midyards, and a single backyard.


Now, let us recall the rebuild process described in \cref{alg:rebuild}. To rebuild the data structure, we first rebuild the frontyard and the midyard part by part, which corresponds to the first three steps. After rebuilding all $P$ frontyards and midyards, we then rebuild the backyard, running \cref{line:lift_backyard}, lifting some of the keys.

For frontyards and midyards, to illustrate the rebuilding procedure, let us assume that we are rebuilding part $i$. Because the individual part is small compared to the entire data structure, we can copy the whole data structure in part $i$ over multiple steps, and then do the rebuild process on only this copy over multiple steps.

One problem is that there may be insert/delete operations on this part during the rebuild time. If an insert/delete operation comes to this part during the copying time, we do the operation normally on the old data structure. If an updated word inside the data structure has already been copied, we also do the same update on the newly copied data structure. In this way, after the copying process is done, the old and new data structures will be identical. 

Now, once we have copied the old data structure, we must also \emph{rebuild} it. As before, there can be some subtleties if an insert / delete operation arrives during this rebuilding phase. If this happens, we perform the operation on the old data structure and then temporarily store it in a buffer. After the rebuild process on the copy is finished, we slowly do the buffered operations on our brand-new data structure. Finally, we replace the old data structure with the new one. 

We summarize the rebuild process for a single part below:

\begin{algorithm}[H]
    \caption{RebuildPart$(i)$}
    Let $\mathrm{Frontyard}_i$, $\mathrm{Midyard}_i$ denote the frontyard and midyard data structures for part $i$. \\
    Let $\mathrm{Frontyard}'_i$, $\mathrm{Midyard}'_i$ be the copies of the data structures that are initialized to be empty. \\
    Let $j = 0$ denote the index of how far along we are with copying elements from $(\mathrm{Frontyard}_i$, $\mathrm{Midyard}_i)$ to $(\mathrm{Frontyard}'_i$, $\mathrm{Midyard}'_i)$. \\ 
    \While{$(\mathrm{Frontyard}'_i$, $\mathrm{Midyard}'_i)$ $\neq$ $(\mathrm{Frontyard}_i$, $\mathrm{Midyard}_i)$}{
    \If{receive operation $O$}{
    Copy $O(1)$ elements from $(\mathrm{Frontyard}_i$, $\mathrm{Midyard}_i)$ to $(\mathrm{Frontyard}'_i$, $\mathrm{Midyard}'_i)$. \\
    Let $j$ be the last element copied. \\
    Execute operation $O$ on appropriate part (if part $i$, execute on $(\mathrm{Frontyard}_i$, $\mathrm{Midyard}_i)$). \\
    \If{$O$ operates on an already copied element}{Execute operation $O$ on $(\mathrm{Frontyard}'_i$, $\mathrm{Midyard}'_i)$.}
    }
    }
    Buffer = $\{\}$. \\
    \While{$(\mathrm{Frontyard}'_i$, $\mathrm{Midyard}'_i)$ not rebuilt}{
    \If{receive operation $O$}{
    Execute operation $O$ on appropriate part (if part $i$, execute on $(\mathrm{Frontyard}_i$, $\mathrm{Midyard}_i)$). \\
    Perform $O(1)$ rebuild steps from \cref{alg:rebuild} on $(\mathrm{Frontyard}'_i$, $\mathrm{Midyard}'_i)$. \\
    \If{$O$ is on part $i$}{
    Buffer $\leftarrow$ Buffer $+ O$.    
    }
    }
    }
    \While{Buffer $\neq \emptyset$}{
    \If{receive operation $O$}{
    Execute operation $O$ on appropriate part (if part $i$, execute on $(\mathrm{Frontyard}_i$, $\mathrm{Midyard}_i)$). \\
    \If{$O$ is on part $i$}{
    Buffer $\leftarrow$ Buffer $+ O$.    
    }
    Execute $O(1)$ operations from Buffer on $(\mathrm{Frontyard}'_i$, $\mathrm{Midyard}'_i)$. \\
    }
    }
    Delete $(\mathrm{Frontyard}_i$, $\mathrm{Midyard}_i)$. \\
    \Return{$(\mathrm{Frontyard}'_i$, $\mathrm{Midyard}'_i)$.}
\end{algorithm}

To bound the size of the buffer, we have the following claim:

\begin{claim}\label{clm:buffernotbig}
    The number of different keys in the operations in part $i$ that occur after copying and before finishing the rebuilding is $O(n/P)$ with high probability.
\end{claim}

\begin{proof}
    The proof is identical to \cref{clm:partnotbig}.
\end{proof}

Note that the claim only bounds the number of different keys, instead of the number of operations because if a single key is inserted and deleted multiple times, all but (at most) one of these operations would cancel out. Thus, for the buffer, we cannot use a simple sequence to store all the operations. Instead, we use a hash table to store all the different keys in these operations with their status (need to be deleted/need to be inserted). Since the order of these insertions / deletions does not matter, we can iterate through the hash table in any order. Thus, by the previous claim, clearing the buffer can also be done within $O(n_i)$ time. Finally, the rebuild time is $O(n_i)$ for part $i$ and can be done within $O(n_i)$ operations with high probability.

The newly introduced storage is the copied data structure and the buffer which cost only \[O\bk*{n_iv + n_i\log\log\bk*{\frac{U}{n_i P}} + n_i\log^{(k)}n_i} + O(n_i\log U) = O(n)\] bits, 
since $n_i\leq O(n/\log^2 U)$ and $v\leq O(\log U)$. 

Finally, we do the rebuild process in the backyard, as stated in \cref{line:lift_backyard}, slowly iterating through the backyard and lifting some of the keys to the frontyard of their corresponding part. This process can be done in $O(n)$ operations. Note that to support this process, we need to be able to fetch a key-value pair from the hash table in $O(1)$ time, delete it, and then insert it into the new hash table. Here, the fetch operation is the second generalization in \cref{thm:dictionarySpace}.

A technical issue is that the operations may only be $O(1)$ time with high probability in the size of the hash table instead of $n$. However, we can easily fix this, since we only use $O(P)=O(\log^2 U)$ hash tables. So, we can simply add some redundant keys to ensure that the size of each hash table is at least $O(n^\epsilon)$. 

\paragraph{Padding the Number of Items}

There is one final subtlety that arises: namely in \cref{thm:dictionarySpace}, there is a restriction on the number of items being always $ \log^{20}(U) \leq n$, and thus as stated, the runtime / space guarantees of \emph{our data structure} would only hold when $n$ is in this same parameter regime. To remedy this, we initialize our data structure with $\log^{20}(U)$ ``sentinel keys'', which ensure that this lower bound on the current number of keys is always being met (and correspondingly, we use a new universe size of $U' = U + \log^{20}(U)$ under the hood). The implication in the space consumption is an additive $O(\log^{20}(U) \cdot \log(U)) = O(\log^{21}(U))$ many bits of space. Importantly, for any $\varepsilon > 0$, $O(\log^{21}(U)) \leq U^{\varepsilon}$, and so we absorb this into our final term. Likewise, the space consumption with universe size $U' = U + \log^{20}(U)$ is unchanged, as $n \log\log(U'/n) = O(n \log\log(U/n))$.

\paragraph{Concluding \cref{thm:mainRestated}}
With this, we establish our main theorem: 

\mainthm*

We remark that throughout this section, we have been assuming that the keys are drawn uniformly at random from the universe. In \cref{sec:hash_function}, we discuss how this assumption can be relaxed using relatively standard techniques with only an extra $U^\epsilon$ bits of space.

{As a corollary of Theorem \ref{thm:mainRestated}, we can construct an extension of the data structure that allows different values to have different sizes, thereby obtaining a very space-efficient method for performing memory allocation on small objects:
\begin{corollary}
Let $U \ge \poly n$ be a parameter, let $k = O(1)$, and let $\eps >0$ be arbitrary constants. There is a \emph{resizable} retrieval data structure maintaining a set of key-value pairs, with keys from the universe $[U]$ and with values of arbitrary sizes in $[U]$ (the size of a value is specified on insertion), such that:
  \begin{enumerate}
      \item At any given moment when it stores $n$ keys $S$, with $n\leq U/2$, it occupies $O(n \log \log (U) + n\log^{(k)}n) + \sum_{s \in S} v_s + U^{\eps}$ bits of space, where $v_s$ denotes the size of the value for key $s$.
      \item It supports insertions, deletions, and retrieval queries in $O(k) = O(1)$ time with high probability in $n$ under the word RAM model with word size $w = \Theta(\log U)$.
  \end{enumerate}
\end{corollary}
\begin{proof}
The construction is the same as Theorem \ref{thm:mainRestated} except that the data structure always stores values in separate arrays from where it stores other information. It can then implement these value arrays with the array construction in Theorem 6 of \cite{jansson2012cram} that stores an array of $m$ values $w_1, \ldots, w_m$, each of length $|w_i| \le O(\log U)$, in space $O(m\log \log U + \operatorname{polylog} U) + \sum_i |w_i|$ bits, while supporting $O(1)$-time operations. 
\end{proof}
}

Finally, we conclude the section with a simple corollary for dynamic filters:
\begin{corollary}\label{cor:resizableFilter}
Let $U$ and $\epsilon$ be parameters satisfying $\log \epsilon^{-1} \le O(\log U)$, and consider the word-RAM model with word size $\Theta(\log U)$. Then, there is a resizable filter with false-positive rate $\epsilon$ that offers the following guarantees in terms of its current size $n$ at any given moment: it supports constant-time operations with high probability in $n$, and it uses space at most
$$n \log \epsilon^{-1} + O\left(n \log \log (U/n)\right) + \poly\log U + O(U^\delta)$$
bits, where $\delta \in (0, 1)$ is a positive constant of our choice and where the $U^\delta$ term is just to store a hash function. 
\label{cor:filter}
\end{corollary}
\begin{proof}
This follows by using Theorem \ref{thm:mainRestated} to store a $v = \log \epsilon^{-1}$-bit value $H(x)$ for each key $x$, where $H$ is a pairwise-independent hash function. The query operation Query$(x)$ simply queries the retrieval data structure with key $x$, takes the output $u$, and returns true iff $u = H(x)$. If $x$ is present, then we get a return-value of $\texttt{true}$ with probability $1$, and if $x$ is not present, we get a return-value of $\texttt{false}$ with probability $1 - \epsilon$. 
\end{proof}

Note that, when $n \ge U^{\Omega(1)}$, Corollary \ref{cor:filter} gives a space bound of 
$$n\log \epsilon^{-1} + O(n \log \log n)$$
bits, improving upon the bound previously given by \cite{PSW13}, which was in terms of the \emph{maximum} number of keys present over all time, rather than the current number $n$.

\section{Lower Bounds}
\label{sec:lower_bound}
\subsection{Space Lower Bound}

In this section, we show a lower bound on the space complexity of any resizable retrieval data structure. Our starting point is the lower bound of Pagh, Segev, and Wieder \cite{PSW13} for resizable filters. 
\begin{theorem}[Theorem 3.1 of \cite{PSW13}]\label{thm:PSW13}
    Consider any filter $\mathcal{D}$ with false-positive rate $\epsilon$ that supports an unknown number $n$ of insertions from a universe $[U]$. Let $n \leq \eps U$ be sufficiently large, and let $1 / \sqrt{n} \leq \alpha < 1$. If, for any sequence of insertions of length $m$ such that $\alpha n < m < n$ the maximum space usage of the data structure $\mathcal{D}$ is $\beta m$ bits of space, then for any integer $\gamma \geq 2$ it holds that 
    \[
    \beta  \geq \left ( 1 - \frac{1}{\gamma}\right ) \cdot (\log(1 / \eps) + (1 - 9 \eps)\log\log_{\gamma}(1 / \alpha) - \Theta(1)).
    \]
\end{theorem}

In particular, we can apply the above theorem with $\alpha = 1 / \sqrt{n}$ and $\gamma = 2^{\sqrt{\log n}}$, and assume $\eps \leq 1/ 10$. This leads to the bound that 
\[
\beta  \geq \left ( 1 - \frac{1}{2^{\sqrt{\log n}}}\right ) \cdot \left (\log(1 / \eps) + \Omega\left (\log\left ( \frac{\log(1/\alpha)}{\log \gamma}\right ) \right) \right )
\]
\[
= \left ( 1 - \frac{1}{2^{\sqrt{\log n}}}\right ) \cdot \left (\log(1 / \eps) + \Omega\left ( \log\log n \right ) \right ).
\]

Note that the above theorem \emph{only applies} when the value $m$ (the number of elements) is not known in advance, and is allowed to take any value in the interval $[\alpha n, n]$.

To get a lower bound for resizable retrieval, we apply the same reduction as in Corollary \ref{cor:filter} to obtain the following lemma:

\begin{lemma}\label{lem:simApproxMember}
    Let $\mathcal{D}$ be a resizable retrieval data structure defined over the universe $[U]$ with $v$-bit values. Then, $\mathcal{D}$ can be used to build a resizable filter over $[U]$ with $\eps = \frac{1}{2^v}$.
\end{lemma}

\begin{proof}
For every element $x \in \mathcal{U}$, we associate with it a random, independent $v$-bit string, which we denote by $v_x$. To build our resizable filter, we use $\mathcal{D}$ to store $v_x$ for each element $x$ that is present. 

To answer a filter query on a key $x$, we query $x$ in $\mathcal{D}$ in order to obtain some string $z$, and we report $x$ to be present iff $z = v_x$. If $z$ is truly present, then $z = v_x$ and we return true with probability $1$. If $z$ is not truly present, then $v_z$ is independent of $z$, and we return present with probability $\Pr[z = v_z] = 1/2^v = \eps$. 
\end{proof}

With this lemma in hand, we can conclude the main theorem of this section:

\begin{restatable}{theorem}{thmlowerinsonly}
    Let $\mathcal{D}$ be a resizable retrieval data structure over a universe $[U]$, and with $v$-bit values for $v \geq 4$, let $n \leq U / 2^v$ be sufficiently large, with $v \leq \mathrm{poly}\log n$. If, for any sequence of insertions of length $m$ such that $\sqrt{n} < m < n$, the maximum space usage of the data structure $\mathcal{D}$ is $\beta m$ bits, then it holds that 
    \[
    \beta m \geq m \cdot \left (v + \Omega\left ( \log\log n \right ) \right ).
    \]
    \label{thm:lowerinsonly}
\end{restatable}

\begin{proof}
    It follows from \cref{lem:simApproxMember} that such a data structure $\mathcal{D}$ can solve the approximate membership problem with $\eps = \frac{1}{2^v} \leq 1/ 10$. From there we invoke \cref{thm:PSW13} with $\eps = \frac{1}{2^v}, \alpha = \frac{1}{\sqrt{n}}, \gamma = 2^{\sqrt{\log n}}$. This leads to the bound that 
    \[
    \beta m \geq m \cdot \left ( 1 - \frac{1}{2^{\sqrt{\log n}}}\right ) \cdot \left (v + \Omega\left ( \log\log n \right ) \right ).
    \]
    To conclude, we must only observe that $\frac{1}{2^{\sqrt{\log n}}} \cdot v  = o(\log\log n)$, and hence this term is negligible in comparison to the $\Omega(\log\log n)$ term. Thus, the lower bound is exactly 
    \[
    \beta m \geq m \cdot  (v + \Omega\left ( \log\log n \right )),
    \]
    as we desire.
\end{proof}

\begin{remark}
    In the regime where $v = O(\log n)$, the above implies that for a choice of $U = \mathrm{poly}(n)$, the data structure must use $mv + \Omega(m \cdot \log\log(n)) = mv + \Omega(m \cdot \log\log(U/n))$ bits of space. This shows that, in at least one regime of parameters, $\Omega(m \cdot \log\log(U/n))$ redundant bits are necessary for resizeable retrieval, \emph{even in just the incremental setting}. Otherwise, without the resizing requirement, the incremental setting was known to have strictly smaller space requirements (see \cite{KPXZZ24}).
\end{remark}

\subsection{Space-time Tradeoffs}
In this subsection, we present space-time tradeoffs for retrieval data structures. To start, we recall the following result:

\begin{theorem}[Theorem 1.3 of \cite{10353191}]
    Consider a dynamic dictionary storing at most $n$ keys from $[U]$, each associated with a value in $[2^v]$, with $nv + \log \binom{U}{n} + O(n\log^{(k)} n)$ bits, where $k\leq \log^* n$, $U \geq 3n$, and $v = 2\log(U/n)+\Theta(\log U)$. Then, in the cell-probe model with word size $w = \Theta(\log U)$, the expected update time is at least $\Omega(k)$. Here, the dictionary supports all operations supported by a retrieval data structure as well as the membership query, which asks whether a key is in the set.
\end{theorem}

With this, we can now conclude our own space-time tradeoff:

\begin{theorem}\label{thm:spaceTime}
    Consider a retrieval data structure storing at most $n$ keys from $[U]$, each associated with a value in $[2^v]$, with $nv + O(n\log^{(k)} n)$ bits, where $k\leq \log^* n, U \geq 3n$ and $v = \Theta(\log n)$. Then, in the cell-probe model with word size $w = \Theta(\log n)$, the expected update time is $\Omega(k)$.
\end{theorem}

\begin{proof}
    Suppose we have such a retrieval data structure with expected update time $S$. Since $U \ge 3n$, we can use this data structure and only insert/query keys that are in $[3n] \subseteq [U]$ to get a retrieval data structure with universe size exactly $3n$. Then, by adding an additional bit array of length $3n$ maintaining membership information, we can build a dynamic \emph{dictionary} with at most $n$ keys, value length $v$, universe size $3n$, which uses $nv + O(n\log^{(k)} n)$ bits of space and has expected update time $S + O(1)$. Thus, by the previous theorem, we have $S=\Omega(k)-O(1)=\Omega(k)$.
\end{proof}

By this theorem, to achieve $O(k)$-time operations with high probability in $n$, our data structure must use at least $nv + O(n \log^{(k)} n)$ bits of space.

Collectively, \cref{thm:lowerinsonly} and \cref{thm:spaceTime} present a clearer picture of the overall space landscape for resizable retrieval data structures. For the first two additional space terms $O(n \log\log(U / n) + n \log^{(k)}(n))$ introduced in \cref{thm:mainRestated}, there are parameter regimes where these terms are necessary. 

This concludes the section.

\bibliographystyle{alpha}
\bibliography{ref.bib}

\appendix
\section{Using a Global Hash Function}
\label{sec:hash_function}
In this section, we remove the assumption that all inserted/queried keys are uniformly random.
In order to do this, we first apply a hash function to the keys to divide the entire universe into $P$ parts (randomly), and then within each part, we apply another hash function to further randomly hash the keys. Since our primary focus is on the complexity within a single part, throughout this section, we denote the original universe by $[U']$, and the number of keys by $n'$. The universe of the $i$-th segment is denoted by $U = U'/P$, and the number of elements in this part is denoted by $n_i$, where these parameters satisfy $\sum_{i \in [P]} n_i = n'$.

As our starting point, we sample a hash function $H': [U']\to [P]$ from a $U^\epsilon$-wise independent hash-function family to determine which part an input key $x$ belongs to. Once an element is in an assigned part, we use another global hash function $H: [U']\to [U]$ to determine the bits assigned to $x$'s key. Thus, for a key $x$, the data structure maintains it in part $H'(x)$ using its hashed key $H(x)$.

Recall that in \cref{clm:partnotbig} and \cref{clm:buffernotbig}, we proved that the number of keys in a part and the size of the buffer will not be too large by using the uniform randomness of keys. Now, we must show that the same claims are true when we use $H'$. The following claim corresponds to \cref{clm:partnotbig}, and a proof of \cref{clm:buffernotbig} can be derived in an identical manner.

\begin{claim}
    Given a set of keys $S$ of size $n'$ and $i\in [P]$, the number of keys $n_i$ in part $i$ is $\Theta(n'/P)$, i.e., $\lvert \{ x \in S \mid H'(x) = i \} \rvert=\Theta(n'/P)$ with high probability in $n'$.
\end{claim}

\begin{proof}
    Let $V_x$ be the indicator variable of the event that key $x\in S$ hashes to part $i$, i.e., $H'(x)=i$. Then $\forall x\in S, \mathbb{E}[V_x]=1/P$. Moreover, $\{V_x\mid x\in S\}$ are $\log^2 U$-wise independent (by the definition of $H'$). By the Chernoff bound with limited independence \cite[Theorem 5]{schmidt1995chernoffhoeffding}, the probability that $\lvert\sum_{x\in S} V_x-\frac{n'}{P}\rvert\geq \frac{n'}{2P}$ is at most $\exp(-\log^2 U')=U^{'-\omega(1)}$ using the fact that $n'\geq \log^{20} U',P=\Theta(\log^2 U')$.
\end{proof}

Next, for the remainder of the discussion, we consider a single part, which we denote by part $0$. Since all the claims we derive for this part are deterministic or hold with high probability in $n_0 \ge \sqrt{n'}$, the correctness for part 0 extends easily to the correctness of the entire data structure. We denote $n_0$ by $n$ for consistency with \cref{sec:algorithm}.

Now, we need to prove the concentration bounds related to the hash function $H$ (i.e., the concentration bounds within a specific part). To do this, we first give the construction of the hash function $H$. Recall that in each part there are $n=\Theta(n'/P)\geq \log^{20}U$ elements. Let $B$ be the minimum power of $2$ that is greater than $U^{\epsilon / 2}$. We sample a hash function $S:[U']\to [B]$ and $B$ separate hash functions $G_i:[U'] \to [U/B]$. Here, we sample $S$ and each $G_i$ from $U^{\epsilon}$-wise independent hash-function families.

The global hash function $H: [U'] \to [U]$ is defined as
\[
H(x)=S(x)\cdot(U/B)+G_{S(x)}(x),
\]
that is, we first use $S(x)$ to hash each key to one of $B$ \defn{buckets}, then we further use separate hash functions $G_i(x)$ to hash all keys in bucket $i$ to the range $[U]$.

By definition, it is clear that $H$ is $U^\eps$-wise independent. In the remainder of this section, we show several desired concentration bounds on the global hash function $H$. We start by showing that each bucket does not contain too many keys.

\begin{claim}
\label{clm:bucketConcentration}
    When $n \geq U^\epsilon$, for each $i\in [B]$, the number of inserted keys that hash to bucket $i$ is at most $\frac{2n}{B}$ with high probability in $U$.
\end{claim}

\begin{proof}
    For symmetry, we only need to show the claim for bucket 0. Let $X_j$ be the indicator variable of the event that the $j$-th key $x$ hashes to bucket 0, i.e., $S(x) = 0$. We know that $\E[X_j]=\frac{1}{B}$ and $X_1, \dots, X_n$ are $\log^2 U$-wise independent.

    By the Chernoff bound with limited independence \cite[Theorem 5]{schmidt1995chernoffhoeffding}, we have
    \[
    \Pr\Bk*{X_1+\cdots+X_n\geq \frac{2n}{B}} \le e^{-\Omega(\log^2 U)}=U^{-\omega(1)}.
    \qedhere
    \]
\end{proof}

With these properties, we are now ready to provide proof of the claims where the random keys are replaced by the hashed keys after the function $H$. Three claims related to $H$ need to be proved again: \cref{clm:sizeOfMidyard}, \cref{clm:slotBoundLog2U}, and \cref{clm:boundLargeBlock}. We start with the simplest one, \cref{clm:slotBoundLog2U}.

\begin{claim}\label{clm:slotBoundLog2UwithH}
    The number of fingerprints in one subtree is $O(\log^2 U)$ with high probability in $U$.
\end{claim}

\begin{proof}
    Similar to the random case, when a fingerprint is in the subtree, its original hashed key must also be in this subtree since the fingerprint is a prefix of the hashed key. Let $X_i$ be the indicator variable of the event that the $i$-th hashed key is in this subtree. Then \[\E[X_i] = 2^{-\ell} = O\left(\frac{1}{n\log^{10}(U/n)}\right).\] Moreover, $X_1, \ldots, X_n$ are $\log^2 U$-wise independent (by the definition of $H$). If the number of fingerprints in the subtree is at least $O(\log^{2} U)$, we must have $X_1 + \cdots + X_n \ge O(\log^{2}U)$. But by the Chernoff bound with limited independence \cite[Theorem 5]{schmidt1995chernoffhoeffding}, this occurs with probability at most $\exp(-\Omega(\log^2 U)) = U^{-\omega(1)}$.
\end{proof}

The following claim corresponds to \cref{clm:sizeOfMidyard}.

\begin{claim}\label{clm:s-1}
    For $n$ different keys from the universe, and letting $n'$ be the number of distinct length-$\l$ prefixes of their hashed keys, then \[n-n'\leq O\left(\frac n{\log^8(U/n)}\right)\] with high probability in $n$.
\end{claim}

\begin{proof}
    We prove it on a case by case basis depending on the current number of keys $n$. If $\log^{20}U\leq n\leq U^{\epsilon}$, then since $H$ is $U^{\epsilon}$-wise independent, all hashed keys are fully independent. Then we can simply invoke \cref{clm:sizeOfMidyard}, and conclude that $n-n'\leq O\bk*{\frac{n}{\log^8(U/n)}}$ with high probability in $n$.

    If $n>U^\epsilon$, since $2^{\ell}\geq n>U^{\epsilon}\geq B$, $B$ buckets uniformly partition the nodes with depth $\ell$ into $B$ parts. We bound the number of fingerprints with length $>\ell$ in each bucket and add them up.
    If two hashed keys are in the same subtree, we call them a collision. Moreover, we say the collision is in the bucket they belong to. Since \cref{clm:bucketConcentration} tells us each bucket contains at most $2n/B$ keys with high probability in $U$, we define $X_i$ as the minimum of $2n/B$ and the number of collisions in the bucket $i(0\leq i< B)$. Then, adopting the same notation as in \cref{clm:sizeOfMidyard}, we have \[X_i=\min\left(2n/B, \sum_{u} s_u(s_u-1)/2\right)\geq \sum_{u} (s_u-1),\] with high probability in $U$ where $u$ iterates through all vertices with depth $\l$ in bucket $i$, and $s_u$ is the number of keys that contain $u$ as a prefix.

    For a bucket $i\in [B]$, we fix its size as $s$. Now, since $G_i$ is $2$-wise independent, the expected number of collisions in the bucket is \[\E[X_i]\leq \frac{s^2}{2^{\ell}/B}\leq \frac{2n}{B}\frac{s}{2^{\ell}/B}.\] Thus, \[\E\Bk*{\sum_{i=0}^{B-1} X_i}\leq \frac{2n}{B}\frac{n}{2^{\ell}/B}=\frac{2n^2}{2^\ell}.\]

    Finally, we bound the sum of $X_i$'s. Since the $G_i$'s are sampled independently, $\frac{X_i}{2n/B}$ are independent random variables which are confined into the interval $[0, 1]$. By a Chernoff Bound, \[\Pr\left[\sum X_i\geq \frac{4n^2}{2^{\ell}}\right]\leq \exp\left(-\Omega\left(\frac{Bn}{2^\ell}\right)\right)=\exp\left(-\Omega\left(\frac{B}{\log^{10}(U/n)}\right)\right)=U^{-\omega(1)}.\]
    Then by $4n^2/2^\ell=O(n/\log^{10}(U/n))$, the proof is complete.
\end{proof}

Next, we prove the limited-independence hash function version of \cref{clm:boundLargeBlock}.

\begin{claim}
    Given $n$ different keys, the number of length-$\l$ strings $u$ such that there are at least $\Omega(\log^5(U/n))$ hashed keys that contain $u$ as a prefix is $O(n^3/U^2)$ with high probability. Moreover, when $U\geq n^{1.01}$, no such $u$ exists.
\end{claim}

\begin{proof}
    When $n\leq U^{\max(1-\epsilon/8, \epsilon, 0.999)}$, \cref{clm:slotBoundLog2UwithH} tells us that, with high probability, the number of hashed keys that contains a fixed length-$\l$ string $u$ as a prefix does not exceed $\log^2 U\leq \min(\epsilon/8,1-\epsilon, 0.001)^5\log^5 U\leq O(\log^{5}(U/n))$. Now we assume $n\geq \max(U^{1-\epsilon/8}, U^{\epsilon}, U^{0.999})$. We have $2^{\ell}\geq n>U^{\epsilon}\geq B$, so $B$ buckets uniformly partition all possible fingerprints $u$ into $B$ parts.

    Consider a fixed $u$. Suppose the bucket it belongs to contains $k$ hashed keys. Let $X_i(1\leq i\leq k)$ be the indicator variable of the event that $u$ is a prefix of the $i$-th hashed key. We know that \[\E[X_i]=\frac{1}{2^\ell/B} = O\left(\frac{B}{n\log^{10}(U/n)}\right)\] and $X_1,\cdots, X_k$ are $\log^2 U$-wise independent. By \cref{clm:bucketConcentration}, $k\leq \frac{2n}{B}$, thus \[\E\left[\sum_{i=1}^k X_i\right]\leq \frac{2n}{B} \cdot O\left(\frac{B}{n\log^{10}(U/n)}\right) = O\left(\frac{1}{\log^{10}(U/n)}\right).\]  By the Chernoff bound with limited independence \cite[Theorem 5]{schmidt1995chernoffhoeffding},  \[\Pr\left[\sum_{i=1}^k X_i\geq \log^{5}(U/n)\right]\leq \exp\left(-\Omega\left(\log^{5}(U/n)\right)\right).\]

    Define $Y_i$ as the number of possible $u$ which belongs to bucket $i$ and is a prefix of at least $O(\log^5(U/n))$ hashed keys. Again, $Y_i$'s are independent random variables. Since the probability that a subtree contains $\log^5(U/n)$ hash keys is bounded above, \[\E[Y_i]\leq \frac{2n}{B}\exp\left(-\Omega(\log^5(U/n))\right) = \frac{2n}{B} \left(\frac Un\right)^{-\Omega(1)}.\] We also have $Y_i\leq 2n/B$, then by Chernoff Bound, \[\Pr\left[\sum_{i=0}^{B-1} Y_i\geq \frac{n^3}{U^2}\right]\leq \exp\left(-\Omega\left(\frac {n^2B}{U^2}\right)\right) = \exp\left(-\Omega\left(U^{\epsilon/4}\right)\right),\] completing the proof.
\end{proof}

Note that in the construction of $H$, if both $S$ and $G_i$ are uniformly random, then $H$ becomes a uniformly random hash function. Thus this also implies the random case:

\boundLargeBlock*

Note that although the hash functions $H$ and $H'$ will introduce some additional collisions, the number of keys in the backyard due to the collisions is bounded by \cref{clm:s-1}. But, there is still an issue that the backyard cannot store the same hashed keys. The way to solve this is to store the original key $x$ instead of the combination of $H'(x)$ and $H(x)$ in the backyard. This is the reason why we need an entire, unified backyard instead of $P$ separate backyards.

Finally, we analyze the additional space and time overhead introduced by the hash functions. By \cite{Siegel04hashfunction}, both $H'$, $S$ and each $G_i$ can be stored using $U^{\delta}$ bits of space, supporting evaluations of $H'(x)$, $S(x)$ and $G_i(x)$ (and thus $H(x)$) in $O(1)$ time, where $\lim_{\epsilon\rightarrow 0}\delta=0$.

\end{document}